\documentclass[format=acmlarge, review=false]{acmart}

\usepackage{acm-ec-17}
\usepackage{booktabs} 

\usepackage[ruled]{algorithm2e}

\usepackage{amsmath,amsfonts,amssymb,bbm} 
\usepackage{subfiles}
\usepackage{enumitem}   

\begin{document}
\title{Stable Matchings in Metric Spaces:
                     Modeling Real-World Preferences using Proximity}  

\author{Hossein Karkeh Abadi}  
\affiliation{%
  \institution{Stanford University}
  \department{Electrical Engineering}}
 \email{hosseink@stanford.edu}
\author{Balaji Prabhakar}
\affiliation{%
  \institution{Stanford University}
  \department{Electrical Engineering and Computer Science}
}
\email{balaji@stanford.edu}

\makeatletter
\newtheoremstyle{mytheorem}
  {3pt}
  {3pt}
  {\itshape}
  {}
  {\itshape\bfseries}
  {.}
  {.5em}
  {\thmname{#1}\thmnumber{\@ifnotempty{#1}{ }#2}%
   \thmnote{ {\the\thm@notefont(#3)}}}
\makeatother
\theoremstyle{mytheorem}
\newtheorem{theoremA}{\textbf{Theorem}}[section]
\newtheorem{lemmaA}{\textbf{Lemma}}[section]
\newtheorem*{theorem*}{\textbf{Theorem}}

\newcommand{\ones}{\mathbf 1}
\newcommand{\reals}{\mathbb{R}}
\newcommand{\integers}{\mathbb{Z}}
\newcommand{\symm}{{\mbox{\bf S}}}  

\newcommand{\nullspace}{{\mathcal N}}
\newcommand{\range}{{\mathcal R}}
\newcommand{\Rank}{\mathop{\bf Rank}}
\newcommand{\Tr}{\mathop{\bf Tr}}
\newcommand{\diag}{\mathop{\bf diag}}
\newcommand{\card}{\mathop{\bf card}}
\newcommand{\rank}{\mathop{\bf rank}}
\newcommand{\conv}{\mathop{\bf conv}}
\newcommand{\prox}{\mathbf{prox}}

\newcommand{\Expect}{\mathop{\mathbb E{}}}
\newcommand{\Prob}{\mathop{\mathbb P}}
\newcommand{\Co}{{\mathop {\bf Co}}} 
\newcommand{\dist}{\mathop{\bf dist{}}}
\newcommand{\argmin}{\mathop{\rm argmin}}
\newcommand{\argmax}{\mathop{\rm argmax}}
\newcommand{\epi}{\mathop{\bf epi}} 
\newcommand{\Vol}{\mathop{\bf vol}}
\newcommand{\dom}{\mathop{\bf dom}} 
\newcommand{\intr}{\mathop{\bf int}}
\newcommand{\sign}{\mathop{\bf sign}}
\newcommand{\pr}{\noindent \bf Proof}
\newcommand{\epr}{\hfill \qed \vspace{8pt} \\}

\newcommand{\cf}{{\it cf.}}
\newcommand{\eg}{{\it e.g.}}
\newcommand{\ie}{{\it i.e.}}
\newcommand{\etc}{{\it etc.}}

\newcommand{\xx}{\mathbf{x}}
\newcommand{\oo}{\mathbf{o}}
\newcommand{\cC}{\mathcal{C}}
\newcommand{\ba}{\mathbf{a}}
\newcommand{\bX}{\mathbf{X}}
\newcommand{\bY}{\mathbf{Y}}
\newcommand{\bM}{\mathbf{M}}

\newcommand{\cI}{\mathcal{I}}
\newcommand{\cP}{\mathcal{P}}
\newcommand{\cV}{\mathcal{V}}
\newcommand{\cN}{\mathcal{N}}
\newcommand{\cE}{\mathcal{E}}
\newcommand{\cB}{\mathcal{B}}
\newcommand{\cA}{\mathcal{A}}
\newcommand{\reds}{\mathcal{R}}
\newcommand{\blues}{\mathcal{B}}
\newcommand{\cM}{\mathcal{M}}
\newcommand{\cW}{\mathcal{W}}
\newcommand{\cS}{\mathcal{S}}
\newcommand{\cU}{\mathcal{U}}
\newcommand{\cO}{\mathcal{O}}

\begin{abstract}
Suppose each of $n$ men and $n$ women is located at a point in a metric space.  A woman ranks the 
men in order of their distance to her from closest to farthest, breaking ties at random. The men 
rank the women similarly.  An interesting problem is to use these ranking lists and find a stable 
matching in the sense of Gale and Shapley.  This problem formulation naturally models preferences 
in several real world applications; for example, dating sites, room renting/letting, ride hailing 
and labor markets.  Two key questions that arise in this setting are: (a) When is the stable matching 
unique without resorting to tie breaks? (b) If $X$ is the distance between a randomly chosen stable 
pair, what is the distribution of $X$ and what is $E(X)$?  These questions address conditions under
which it is possible to find a unique (stable) partner, and the quality of the stable matching in terms
of the rank or the proximity of the partner. 

We study dating sites and ride hailing as prototypical examples of stable matchings in discrete and 
continuous metric spaces, respectively.  In the dating site model, each man/woman is assigned to a 
point on the $k$-dimensional hypercube based on their answers to a set of $k$ questions with binary 
answers ({\it e.g.}\ , like/dislike). We consider two different metrics on the hypercube: Hamming 
and Weighted Hamming (in which the answers to some questions carry more weight).  Under both metrics, 
there are exponentially many stable matchings when $k = \lfloor\log n\rfloor$.  There is a unique 
stable matching, with high probability, under the Hamming distance when $k = \Omega(n^6)$, and under 
the Weighted Hamming distance when $k > (2+\epsilon) \log n$ for some $\epsilon>0$.  Furthermore, 
under the Weighted Hamming 
distance, we show that  ${\log(X)}/{\log(n)} \to -1$, as $n \to \infty$, when $k > (1+\epsilon) \log n$ 
for some $\epsilon>0$.  In the ride hailing model, passengers and cabs are modeled as points on the 
line and matched based on Euclidean distance (a proxy for pickup time).  Assuming the locations of 
the passengers and cabs are independent Poisson processes of different intensities, we derive bounds 
on the distribution of $X$ in terms of busy periods at a last-come-first-served preemptive-resume
(LCFS-PR) queue. We also get bounds on $E(X)$ using combinatorial arguments. 
\end{abstract}

\maketitle
\newpage
\section{Introduction}
\label{intro}

The stable marriage problem was first introduced by \citet{gale1962college} as a way of 
modeling the college admissions process, in which students are matched with colleges, and 
the process of courtship leading to marriage, in which women and men are matched.  They 
introduced two key properties of matchings: stability and optimality.  These properties 
are quite well-known and we will recall them formally later; for now, we proceed informally. 
Stability captures the requirement that a matching should not pair a man M and a woman W 
with partners whom they {\em both} prefer less than each other.  Should this happen, M and 
W are both incentivized to break up with their assigned partners and match with each other.  
\citet{gale1962college} show that there is always at least one stable matching and present
the {\em deferred-acceptance algorithm} for finding it.  Optimality refers to the quality 
of a matching in terms of the rank of men in their partners' preference lists, and vice versa.  
The stable marriage problem has also been studied in several other real world settings.  One 
famous example is the National Resident Matching Program (NRMP) \cite{roth1984evolution,roth1996nrmp} 
where medical school students are matched to residency programs through a centralized stable 
matching mechanism.  Other examples include online dating 
\cite{hitsch2010matching}, sorority rush \cite{mongell1991sorority}, and school choice \cite{abdulkadiroglu2003school}.

\citet{knuth1976} initiated the theoretical analysis of large-scale instances of the stable 
marriage problem under the ``random preference list'' assumption, where the preference lists 
of each man and woman are drawn independently and uniformly from the set of all permutations.  
Knuth poses the question of estimating the average number of stable matchings when $n$, the
number of men (equal to the number of women) grows large, and provides an integral formula 
for the probability that a given matching is stable.  \citet{pittel1989average, pittel1992likely} 
evaluated this integral and showed that the average number of stable matchings is asymptotic to 
$e^{-1}{n\ln n}$ as $n\to\infty$, and that any given woman (or man) has $\Theta(\log n)$ stable 
partners, on average.  We mention a few other results under the random preference list assumption 
relevant to our work: \citet{immorlica2005marriage} proved that if the preference list of each 
woman has only a constant number of entries, then the number of people with multiple stable partners 
is vanishingly small.\footnote{\citet{roth1999redesign} also empirically observe this phenomenon 
in the context of candidates interviewing for jobs.}  \citet{ashlagi2014unbalanced} 
studied the ``unbalanced'' case when there are $n$ men and $n+k$ women, for $k \ge 1$.  They show 
that, with high probability,\footnote{We say a sequence of events $E_n$ occur with high probability
 if $\lim_{n \to \infty} \Prob(E_n) = 1$.} the fraction of men and women with multiple stable 
partners tends 
to zero as $n\to\infty$.  This line of work is theoretically very interesting, but preference 
lists in the real world are rarely drawn at random---there can exist a significant correlation 
in the choices people and organizations make.  For example, \citet{roth1999redesign} empirically 
observed correlations in the NRMP preference lists; the applicants largely prefer the same programs 
and the programs tend to rank the applicants similarly ({\it i.e.}, a top applicant in one program was 
also top-ranked in other programs).  They note that these correlations can result in a small set 
of stable matchings.  \citet{holzman2014matching} make the previous observation mathematically
precise by assuming each participant picks their preference list from a small set of permutations. 

While the above correlations capture a ``sameness'' in the preferences of people and organizations,
in this paper we consider correlations due to ``proximity''.  Proximity can arise from a coincidence 
of likes and interests between members of the two sides of a matching market.  For example, each member 
of a matching market answers a questionnaire describing their likes, dislikes, interests or requirements.  
The questionnaire can either be the same for both sides of the matching market ({\it e.g.}, online dating) or 
different ({\it e.g.}, renters answer questions describing their preferred properties while lessors describe 
attributes of their preferred renters).  The vector of answers can be viewed as points in a metric space 
and proximity is equated with distance in the metric space.  Each participant in the market ranks members 
of the other side based on their proximity to the participant, from closest to farthest.  Distance also 
arises naturally in the case of ride hailing, where it is desirable to match a hailer with the closest 
available car.  Thus, a wide variety of real world applications can be modeled in this framework; for 
example, dating sites,\footnote{Tinder (\url{https://www.gotinder.com}), Zoosk (\url{https://www.zoosk.com})} 
renting/letting,\footnote{Airbnb (\url{https://www.airbnb.com}), Zillow (\url{http://www.zillow.com})}    
labor markets,\footnote{LinkedIn (\url{https://www.linkedin.com})}
and ride hailing.\footnote{Uber (\url{https://www.uber.com}), Lyft (\url{https://www.lyft.com})}

\smallskip
\noindent\textbf{Our results.} We analyze stable matchings in discrete and continuous metric spaces as the number 
of participants grows large.  
We make distributional assumptions on the distances between the participants (hence on the preference lists) 
and analyze the number and quality of stable matchings.  The quality of a stable matching is captured by how
small the distances are between stable partners in the matching.  When the metric space is continuous, the 
stable matching is almost surely unique under very mild and natural distributional assumptions.  However, this 
is not necessarily true in discrete metric spaces.  An interesting finding of our work is that a participant 
(on either side of the market) is {\em at the same distance} from their partner in all stable matchings.  Thus,
it makes sense to consider $X$, the distance between a randomly chosen stable pair (regardless of which 
stable matching they're picked from, should there be more than one stable matching).  We are interested in 
the distribution of $X$ and $E(X)$ as the number of participants grows large.  We explore these quantities
in the dating sites and ride hailing settings.

\noindent\textbf{Dating sites}. 
Suppose the men and women of a community are seeking to get matched to a partner in a dating site.  
At the time of signing up, participants are usually asked to answer a fixed set of $k$ yes/no questions 
about their preferences, (\emph{e.g.}, ``Do you like pets?'', ``Are you a morning person?'').  
We call the $k$-bit vector representing a participant's answers to these questions the participant's {\it profile}.
Each profile can be modeled as a point on the $k$-dimensional hypercube, $Q_k$.  The aim
is to match a woman to a man whose profile is closest or most similar to hers.  We consider two different
metrics on $Q_k$ for measuring this similarity: the {\it Hamming} distance and the {\it Weighted Hamming} distance. 
The Hamming distance between two profiles is equal to the number of entries at which they disagree.  The 
Weighted Hamming distance weighs some disagreements more; the details are in Section \ref{discrete}.  Since 
the distances are not necessarily distinct, we also assume that each person has a ``tie-breaking preference list'' for ranking members of the other side and uses this to break ties.  One way to think of the actual preference
list of a woman is that it ranks the men by distance, closest first.  Men at the same distance are ranked according
to her tie-breaking preference list. The men form their preference lists similarly.\footnote{One way to generalize this model to matching markets with two 
 different questionnaires (one for each side of the market) is to ask 
each participant to answer their questionnaire and also to indicate their best answers 
from participants on the other side of the market ({\it e.g.}, renters and lessors answer their
questions and that of an ideal response from the other side).   The overall profile is then 
formed by concatenating the answers to both questionnaires.}

We consider the setting in which profiles are picked independently and uniformly at random from $Q_k$, and 
the tie-breaking preference lists are chosen independently and uniformly from the set of all permutations. 
Let $\epsilon > 0$ be an arbitrary positive number.  We shall prove that under both the Hamming and the Weighted
Hamming distances, for $k < (1-\epsilon) \log n$, the fraction of people with multiple stable partners tends to 
zero, with high probability, as $n \to \infty$.  However, if $k = \lfloor \log n \rfloor$, there are exponentially 
many stable matchings.  We show that, with high probability, the stable matching is unique under the Hamming 
distance for $k = \Omega(n^6)$, and it is unique under the Weighted Hamming distance for $k > (2+\epsilon) \log n$, without resorting to tie breaks.\footnote{Tie-breaking represents {\em chance}, which, in the context
of dating, could reasonably be thought of as being less preferable to {\em choice}.  In other words, a participant would prefer to find his/her partner from their profile rather than through a process involving a coin flip. }
We derive a lower bound on $X$ under the Hamming distance.  Under the Weighted Hamming distance, we prove that if 
$k > (1+\epsilon) \log n$, then $\log X/ \log n \to -1$ in probability.

\noindent \textbf{Ride hailing}.  Consider the problem of matching passengers and cabs on a street.  
Let blue and red points on the real line represent the location of passengers and cabs, 
respectively.  Suppose the blue and red points occur according to two independent Poisson processes 
with respective intensities $\lambda$ and $\mu$.  Each point forms its preference list by ranking 
points of the other color in an increasing order of their Euclidean distance to it.  \citet{holroyd2009poisson} 
studied translation-invariant matchings between the points of two $d$-dimensional Poisson processes with 
the same intensities ($\lambda=\mu$).  They show the natural algorithm of matching mutually closest pairs
of points iteratively yields an almost surely unique stable matching.  They analyze the tail behavior of 
$X$, the distance between a typical pair of stable partners.  In the 1-dimensional case, they derive power 
law upper and lower bounds for the tail distribution of $X$.  In this paper, we study the stable matching 
problem between two Poisson processes on the real line in the unbalanced case where $\lambda < \mu$.  
We derive bounds on the distribution of $X$ in terms of the busy cycles of a last-come-first-served preemptive-resume (LCFS-PR) 
queue.\footnote{Such a queue is also called a {\em stack} \cite{kelly2014stochastic}.}  Using combinatorial arguments, we
prove that $\Expect(X) \le \big{(}1+\ln \frac{\mu+\lambda}{\mu - \lambda} \big{)}/{(\mu-\lambda)}$.

The rest of the paper is organized as follows. In Section \ref{prelims} we define the stable matching
problem, introduce relevant notation, and state some known results.  In Section \ref{discrete} we describe 
the stable matching problem on hypercubes and present our results in this model.  In Section \ref{continuous} 
we analyze the stable matching problem on the real line.  Section \ref{conclusion} concludes the paper.

\section{Background and Previous Work}
\label{prelims}

A community of $n$ men and $n$ women is represented by sets $\cM$ and $\cW$, respectively. 
Suppose each person $x$ in the community has a strict preference list, $\succ_x$, which ranks
members of the opposite gender.  Thus, $y_1 \succ_{x} y_2$ means $x$ prefers $y_1$ to $y_2$. 
A matching $\mu$ is a mapping from $\cM \cup \cW$ to itself, such that for each 
man $m$, $\mu(m) \in \cW \cup \{m\}$, for each woman $w$, $\mu(w) \in \cM \cup \{w\}$, 
and for any $m,w \in \cM \cup \cW$, $\mu(m) = w$ implies $\mu(w) = m$. 
A man or woman $x$ is {\it unmatched} under $\mu$ if $\mu(x) = x$. 
A pair $(m,w) \in \cM \times \cW$ is called a {\it blocking pair} for a matching 
$\mu$ if $w \succ_m \mu(m)$ and $m \succ_w \mu(w)$. 
A matching is called {\it stable} if it does not have any blocking pairs. If a man 
$m$ and a woman $w$ are matched to each other in a stable matching, 
we say $w$ and $m$ are a {\it stable partner} of each other. 

The problem of stable matching was first introduced by \citet{gale1962college}.  They proved 
that there always exists a stable matching, which can be found using an iterative algorithm 
called the {\it deferred-acceptance algorithm}. This algorithm proceeds in a series of proposals
and tentative approvals until there is a one-to-one matching between the men and women.  When the
women propose, they each end up with the best stable partner they can have in any stable matching.  
This matching, often called {\em woman-optimal}, also pairs each man with his lowest-ranked stable
partner.  The {\em man-optimal} stable matching, which results when the men do the proposing, may 
be distinct from the woman-optimal stable matching; thus, there may be many stable matchings.
Under the random preference list assumption, \citet{pittel1989average, pittel1992likely} proved that
the average number of stable matchings is asymptotic to $e^{-1}n\ln n$ as $n\to\infty$, and each person 
has $\Theta(\log n)$ stable partners, on average.
The stable marriage problem can be extended to the unbalanced case where the number of men and women 
is not equal.  It is clear that for any stable matching in the unbalanced case, there are some people 
who remain unmatched.  This may also happen in the balanced case if the preference lists of some men 
or women are not complete.  We state the following theorems for ready reference.

\begin{theorem}[Rural Hospital] \cite{roth1986allocation, mcvitie1970stable}
The set of men and women who are not matched is the same for all stable matchings.
\end{theorem}
 
\begin{theorem}\cite{immorlica2005marriage} \label{immorlica}
Consider the stable marriage problem with $n$ men and $n$ women.  Suppose the preference lists 
of the women are drawn independently and uniformly at random from the set of all orderings of men. 
For a fixed $k\geq 1$, let the preference lists of the men be drawn independently and uniformly 
at random from the set of all ordered lists of any $k$ women.  (The $k$ women on two different
men's preference lists may be different.)  In this setting, the expected number of women who have 
multiple stable partners is $o(n)$.
\end{theorem} 

\begin{theorem}\cite{ashlagi2014unbalanced} \label{ashlagi}
 Consider a stable marriage problem with $n$ men and $n+k$ women, for arbitrary $k = k(n) \ge 1$.  
 Suppose the preference lists of women are drawn independently and uniformly at random from 
 the set of all orderings of men, and the preference lists of men are drawn independently and 
 uniformly at random from the set of all orderings of women.  The fractions of men and women 
who have multiple stable partners tends to zero, with high probability, as $n\to \infty$.
\end{theorem}

\noindent 
The independence of the randomly drawn preference lists is the key assumption in the analysis of 
both Theorem \ref{immorlica} and Theorem \ref{ashlagi}. Under this assumption, Theorem \ref{immorlica} shows 
that if the preference lists of one side of the market is limited to a fixed $k \ge 1$ entries, the fraction
of men and women with multiple stable partners is vanishingly small. Theorem \ref{ashlagi}
proves the same result for unbalanced markets where there is a size $k \ge 1$ discrepancy  
between the number of men and women. In the following section, we derive similar results
for the matching markets with correlated preference lists where each person reveals $k \ge 1$
bits of information about their preference by answering $k$ yes/no questions.

\section{Stable Matching on Hypercubes} \label{hypercubesection}
\label{discrete}

Consider a dating site with $n$ men and $n$ women, represented by sets $\cM$ and $\cW$.  
Let $\cS = \cM \cup \cW$ and let $k$ be a positive integer.  For each $x\in \cS$, let the 
$k$-bit vector representing their profile be denoted by
$\ba_k(x) = \bigl(a_1(x), \dots, a_k(x)\bigr) \in \{0,1\}^k$, 
where $a_i(x) = 0$ if $x$'s answer to the $i^\text{th}$ question is ``no'', and $a_i(x) = 1$ 
otherwise.  Thus, each profile is a point on the $k$-dimensional hypercube, $Q_k = \{0,1\}^k$.
For simplicity, we shall suppress the subscript $k$ from $\ba_k$ whenever it can be inferred.

In this setting, participants prefer to be matched to someone with a similar profile.
Similarity is measured using two metrics on $Q_k$: The  Hamming distance and the Weighted 
Hamming distance.
The Hamming distance $d_h(\ba, \ba')$ between $\ba$ and
$\ba'$ equals
\begin{align}
d_h(\ba, \ba') \triangleq \sum_{i=1}^k \ones(a_i\neq a_i'). \nonumber
\end{align}
The Hamming distance assumes that all questions have the same weight.  However, some questions 
may have higher importance than others.  For example, "Are you allergic to cats?" will likely 
outweigh "Do you like caramel?". The Weighted Hamming distance, 
\begin{align}
d_w(\ba, \ba') \triangleq \sum_{i=1}^k {2^{-i}}\ones(a_i\neq a_i'), \nonumber
\end{align}
addresses this by assigning different weights to different questions.  

\smallskip

\noindent {\bf Remark.} Our results for the Weighted Hamming distance (Theorem \ref{weightedHamming})
can be extended to any exponentially decaying weights. 

\smallskip
\noindent{\bf Remark.}  When making statements which apply to both metrics we shall use the notation 
$d(.,.)$.  We shall use $d(x,y)$ to denote the distance between the profiles of participants
$x$ and $y$.  

\smallskip
\noindent The preference list of $x$ is arranged according to distance, as follows:
for $x, y, y' \in \cS$,
\begin{align}
y \succeq_x y'\quad \Longleftrightarrow \quad d(x, y) \le d(x, y').
\nonumber
\end{align}
Since distances are not necessarily distinct, a tie-breaking rule is needed to strictly order 
preference lists.  As mentioned in the Introduction, participant $x$ uses their ``tie-breaking list'',
$T_x$, to break ties.  Thus, each woman $w$, ranks men in increasing order of their distance to her 
and arranges men at the same distance according to their order in her tie-breaking list, 
$T_w$.\footnote{If $w$ breaks ties at random, then $T_w$ is a random ordering of all the men.}   
For any $x$ and $y$ in $\cS$, $T_x$ is not necessarily equal to $T_y$.  Let the final {\em strict} 
preference list of user $x$ be denoted by $P_x$.  We shall use $\succ_x$ to indicate ordering in this 
list.  We are now ready to state 

\smallskip

\noindent{\bf The Profile Matching Problem (PMP).}  Given $n$ men and $n$ women and their strict preference
lists, the profile matching problem seeks to find a stable matching between the men and the women.

\smallskip
A priori, it seems there may be many stable matchings and
multiple stable partners for some women and men.  However, we shall see in Lemma \ref{sameDistanceLemma} 
that the multiple stable matchings, should they exist, are all essentially equal in quality.
Suppose $\mu$ is a stable matching for the PMP.  Let $d_\mu(x)$ be the distance between $x$ and $\mu(x)$, 
 \begin{align}\label{eq1}
 d_\mu(x) \triangleq d\big(x, \mu(x)\big). \nonumber
\end{align}

\begin{lemma}\label{sameDistanceLemma}
Let $\mu_1$ and $\mu_2$ be two stable matchings for the Profile Matching Problem. 
Then $d_{\mu_1}(x) = d_{\mu_2}(x)$ for every $x\in\cS$.
\end{lemma}

\begin{proof}
See appendix \ref{hypercubeappendix}.
\end{proof}

\noindent According to Lemma \ref{sameDistanceLemma}, $d_\mu(x)$ does not depend on $\mu$.  Hence, we shall 
simply denote $d_\mu(x)$ by $d(x)$ and call $d(x)$ the {\it matching distance} of $x$.  Let 
the random variable $X$ denote the matching distance of a randomly chosen participant $x$. We analyze $X$
in the following section.

\subsection{ The Random Profile Matching Problem (RPMP)}  
We now analyze the PMP under certain distributional assumptions of preference lists and profiles
when the number of participants grows large.  Our main goals are to understand the following 
questions:  How many questions are needed to find a unique partner for each participant without 
resorting to tie-breaking? What is the matching 
distance, $X$?  These questions will be answered under the Hamming and the Weighted Hamming metrics.    

\smallskip

\noindent{\bf Probabilistic assumptions.}  We assume each participant answers each of the $k$
questions equally likely with a ``yes'' or a ``no''.  Further, the answers to all questions by
all the participants are independent.  Geometrically, this assumption places the $k$-bit profile
vector of each participant (or, equivalently, the participant) at one of the $2^k$ vertices of 
$Q_k$, independently and uniformly at random.  The preference lists are then generated based
on the distances induced by the above placement and the tie-breaking lists $T_x, x\in\cS$.  
We assume each $T_x$ is generated independently and uniformly at random from the set of
all orderings of men (or women, depending on $x$).  

\smallskip

\noindent{\bf The RPMP-$k$.}  Given $n$ men and $n$ women, each of whose preference lists are generated
according to the above probabilistic assumptions, the RPMP-$k$ aims to find a stable matching between
the men and the women.

\smallskip

\noindent{\bf Remark.}  Note that RPMP-0 is equivalent to the standard stable matching problem with 
randomly generated preference lists.

\subsection{{Our results}}
 In this section we present our main results for the RPMP-$k$. Theorem \ref{smallk} considers 
 the case where $k \le \lfloor \log n \rfloor$ and Theorem \ref{hamming} and Theorem \ref{weightedHamming}
 study larger values of $k$. Due to page limitation we moved all the proofs to appendix. 
 
\begin{theorem} \label{smallk}
Consider the RPMP-$k$ for $k \ge 1$.  
Fix $\epsilon > 0$.  Under any metric on $Q_k$, the following statements hold 
with high probability:
\begin{enumerate}[label=(\roman*)]
\item if $k < (1-\epsilon)\log n$, the fraction of users with multiple stable partners
tends to zero as $n \to \infty$, so long as tie-breaking is used; and
\item if $k = \lfloor \log n \rfloor$, there are $O(n)$ users with multiple stable partners 
and there are exponentially many stable matchings. 
\end{enumerate}
\end{theorem}

\begin{theorem} \label{hamming}
Under the Hamming distance, with high probability, we have the following:
\begin{enumerate}[label=(\roman*)]
\item if $k(n) = \Omega(n^3)$, the fraction of users with multiple stable partners 
tends to zero as $n \to \infty$;
\item  if $k(n) = \Omega(n^6)$, the stable matching is unique without resorting to tie-breaking; and
\item for any $\beta > 1$, 
\begin{align}
\Prob\left( X < \frac{k}{2} - \sqrt{\beta k\log n} \right) \le n^{1-\beta}. \nonumber
\end{align}
\end{enumerate}
\end{theorem}

\begin{theorem} \label{weightedHamming}
Fix $\epsilon > 0$. Under the Weighted Hamming distance, with high probability, we have the following:
\begin{enumerate}[label=(\roman*)]
\item if $k(n) > (1+\epsilon) \log n$, the fraction of users with multiple stable 
partners tends to zero as $n \to \infty$. Moreover,
\begin{align}
\frac{\log {X}}{\log n} \xrightarrow{~p~} -1, \nonumber
\end{align}
where $\xrightarrow{~p~}$ represents convergence in probability; and
\item if $k(n) > (2+\epsilon) \log n$, the stable matching is unique, without resorting to 
tie-breaking. 
\end{enumerate}
\end{theorem}

\smallskip
According to Theorem \ref{smallk}, in large instances of the RPMP-$k$, if users answer even one question 
($k=1$), the preference lists become skewed so that, with high probability, any given participant has a 
unique stable partner.  This contrasts starkly with the case $k=0$, where \citet{pittel1992likely} showed 
that each participant has, on average, $\Theta(\log n)$ stable partners.  In Theorem \ref{hamming} and
\ref{weightedHamming} we distinguish the statements ``the fraction of participants with 
a unique stable partner goes to 1 with high probability'' from the statement ``there is a unique stable 
matching'', since the former does not imply the latter.  Moreover, our method of proving the latter consists 
of proving the following two steps: (i) if the distances of each man from a given woman are distinct, then 
she will have a unique stable partner (see Lemma \ref{uniqueStablePartner2}); and (ii) if this holds for all the 
women (or all the men), then the stable matching is unique. From a market design perspective the uniqueness 
of the stable matching is important to achieve a shape prediction of the market. Theorem \ref{hamming} shows 
that under the Hamming distance, if $k(n)=\Omega(n^6)$, with high probability, there exists a unique stable 
matching without resorting to tie-breaking. However, asking that many questions from users is not feasible. 
On the bright side, Theorem \ref{weightedHamming} shows that if the answers to questions carry different
weights, we can achieve a unique stable matching with 
$k(n) = O(\log n)$ questions. 

These theorems also study the matching distance, $X$. It will be clear from the proof of Theorem \ref{smallk} 
that $X=0$, with high probability, when $k < (1-\epsilon)\log n$. Theorem \ref{hamming} establishes an upper bound on 
the matching distance $X$. Theorem \ref{weightedHamming} covers the case of the Weighted Hamming metric. 

\smallskip
\noindent \textbf{Remark.} All above theorems can be extended to unbalanced markets with $n$ men and
$n+r$ women.

\section{Stable Matching On the Line} \label{continuous}
Consider the problem of matching passengers and cabs on a street.  Suppose the passengers 
and cabs are represented as blue and red points, respectively, on $\reals$.  Let $\blues$ 
and $\reds$ denote the set of blue and red points, respectively.  Let $\cS = \blues \cup \reds$.  
A matching between $\blues$ and $\reds$ is a mapping $\cM$ from $\cS$ to $\cS \cup \{\infty\}$, 
such that for every red point $r$, $\cM(r) \in \blues \cup \{\infty\}$, for every blue point 
$b$, $\cM(b) \in \reds \cup \{\infty\}$, and for every $b, r \in \cS$, $\cM(r) = b$ implies 
$\cM(b)=r$. A point $x\in \cS$ is unmatched if $\cM(x) = \infty$.  The preference list of 
each point is based on its Euclidean distance to the points with a different color, closest
first.  A matching $\cM$ is stable if there is no pair $(b,r) \in \blues \times \reds$ such that
\begin{align}
b \not{=} \cM(r) ~~\textnormal{and}~~ |r - b| < \min\bigl(|r-\cM(r)|, |b-\cM(b)|\bigr).
\nonumber
\end{align}
For any matching $\cM$ and any point $x \in \cS$, 
let $I_\cM(x) \subset \reals$ denote the open interval which has $x$ and $\cM(x)$ at its 
end-points, and let $d_\cM(x)$ represent the 
length of $I_\cM(x)$, {\it i.e.}, $d_\cM(x) = |x - \cM(x)|$. We call $I_\cM(x)$ the
{\it matching segment} of $x$, and $d_\cM(x)$ the {\it matching distance} of $x$ in $\cM$.

With the above definitions, suppose that points in $\blues$ and $\reds$ occur according to independent 
Poisson processes with rates $\mu$ and $\lambda$, respectively, where $\lambda \le \mu$.  We call the 
matching problem defined above as the {\it Poisson Matching} problem and denote it by $PM(\lambda, \mu)$.
As mentioned in the Introduction, \citet{holroyd2009poisson} studied translation-invariant matchings between 
two $d$-dimensional Poisson processes with the same intensities; in particular, they studied stable matchings.
They showed that the following algorithm finds a unique stable matching: Each blue point simultaneously emits 
two rays, one in each direction, such that at any time $t$, each ray is at distance $t$ from its emitter.
Once a ray hits an unmatched red point $r$, the emitter $b$ will be matched to $r$, and both points leave 
the system.  Denote the unique stable matching by $\cM_s$ and let $x\in \blues$ be an arbitrary blue point. 
Define the random variable $X$ to be $x$'s matching distance in $\cM_s$, {\it i.e.}, 
$X \triangleq d_{\cM_s}(x)$.  \citet{holroyd2009poisson} proved that if $\mu = \lambda$, $\Expect (X^{1/2}) = \infty$.

\begin{theorem}\cite{holroyd2009poisson} \label{holroydTheorem}
Let $\blues$ and $\reds$ be independent 1-dimensional Poisson processes of intensity $1$, and 
let $X$ represent the matching distance of an arbitrary point in the stable matching between 
$\blues$ and $\reds$. We have,
\begin{align}
\Expect( X^{1/2}) = \infty \quad \text{and} \quad \Prob(X>r ) \le Cr^{1/2} \quad \forall r > 0,
\nonumber
\end{align}
for some constant $C\in (0, \infty)$.
\end{theorem}

\noindent In this section we analyze the 1-dimensional $PM(\lambda,\mu)$ problem for $\lambda < \mu$.  This models the situation in 
which there are fewer passengers than cabs and sheds light on the time it would take for a passenger to be
picked up by the nearest cab that is assigned to pick up the passenger.\footnote{Note the nearest cab may 
not be able to pick up a passenger since it may be assigned to pick up another passenger who is nearer to 
the cab than the first passenger.  Hence, stable matchings are quite natural in this setting.}  Thus, we
shall be interested in the distribution (Theorem \ref{bound}) and the expected value (Theorem 
\ref{Xupperbound}) of $X$.  However, in order to get at these quantities, we need to introduce various
ideas such as the relationship among $PM(\lambda,\mu)$, last-come-first-served preemptive-resume 
(LCFS-PR) queue, and nested matchings.  We believe these ideas are interesting in their own right.

\subsection{Queue Matching}
Red partners in a stable matching may be either to the left or to the right of the corresponding
blue points.  However, in queue matchings they are either only on the left or only on the right.
Consider $PM(\lambda,\mu)$ with the constraint that each blue point can only be matched to 
red points that are on its {\em right}.  In the passenger-cab scenario, this constraint can 
be the result of having a one-way street or a road divider, where each cab can only pick up 
passengers on its left.  In order to find the stable matching, all the blue points simultaneously
emit a ray to their right at time 0.  Once a ray hits an unmatched red point $r$, the emitter $b$ 
will be matched to $r$.  It is clear that this algorithm is equivalent to running an LCFS-PR queue 
where the time of job arrivals and departures in this queue are represented as blue points and red 
points, respectively.  The arrival rate is $\lambda$ and the service rate is $\mu$ (the service times
are i.i.d.\ exponentials of rate $\mu$). We call the resulting stable matching, $\cM_s^+$, the {\it forward queue} 
matching, corresponding to running the queue forward in time.  Similarly, we can define a {\it backward 
queue} matching, $\cM_q^-$, where each blue point is matched to a red point on its left, and can be 
found by running the LCFS-PR queue backward in time.  Figure $1$ shows $\cM_s^+$,
$\cM_s^-$, and $\cM_s$ for an instance of the problem.
\begin{figure}[t]
\begin{center}
\includegraphics[width=.45 \textwidth]{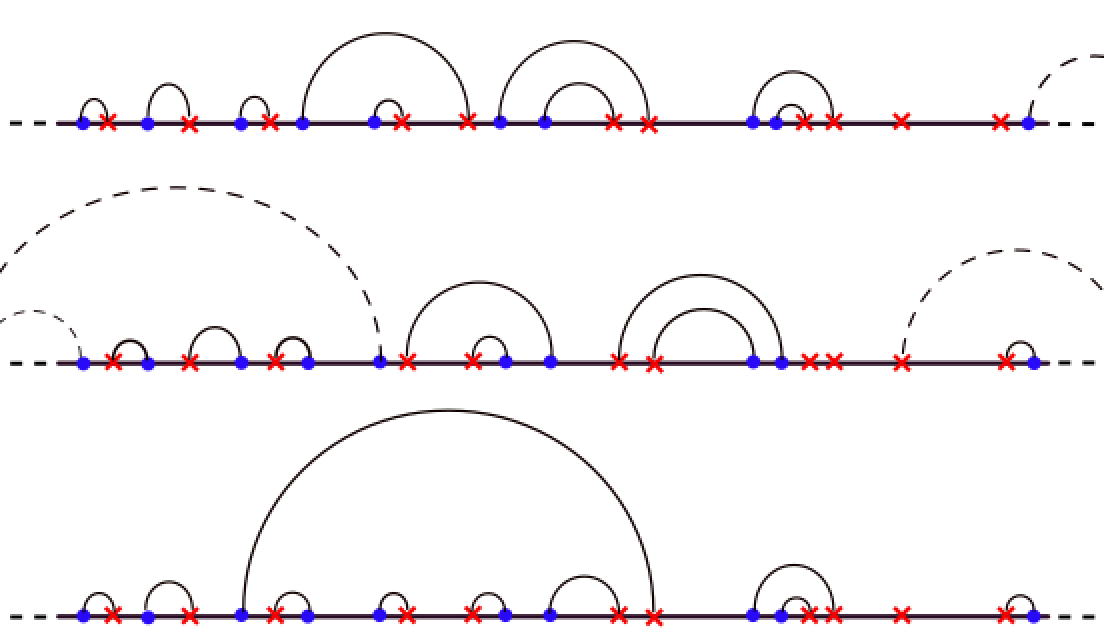}
\end{center}
\caption{Different matchings between two point processes.
Top: forward queue matching $\cM_s^+$,
Middle: backward queue matching $\cM_s^-$,
Bottom: stable matching $\cM_s$.}
\label{f-nopsfrag}
\end{figure}
The following are well-known facts about LCFS-PR queues with rate $\lambda$ Poisson arrivals and  
rate $\mu>\lambda$ i.i.d.\ exponential service times which are independent of the arrival process.  
Since $\lambda < \mu$, the queue is stable and each blue point in $\blues$ almost surely has a 
partner in $\reds$.  Let $x\in \blues$ be an arbitrary blue point and let $X^+$ be $x$'s matching 
distance in $\cM_s^+$.  It is clear that $X^+$ has the same distribution as the busy cycle in the
corresponding LCFS-PR queue, where the {\it busy cycle} is the duration of time from the arrival
of a job at an empty queue to the time the job leaves the queue.  It is known \cite{gross1998fundamentals} 
that the probability density function of the busy cycle is given by
\begin{align}
f_{\lambda, \mu}(t) = \frac{1}{t\sqrt{\rho}}e^{(\lambda + \mu)t}I_1(2t\sqrt{\lambda\mu}),
\quad t>0,
\nonumber
\end{align}
where $\rho = \frac{\lambda}{\mu}$, and $I_1$ is the modified Bessel function of the first kind. 
Let $B(\lambda, \mu)$ represent this distribution.
The average busy cycle duration is $\Expect(X^+) = 1/(\mu - \lambda)$.  In the following section we 
introduce a class of matchings which includes both stable and queue matchings.  
 
\subsection{Nested Matching}
For any interval $I\in \reals$, represent its closure by $\bar{I}$.  A matching $\cM$ is said 
to be \emph{nested} if for any $x,y \in \cS$, $x \in I_\cM(y)$ implies 
$\cM(x) \in \overline{I}_\cM(y)$.  Therefore, in any nested matching if $I_\cM(x) \cap I_\cM(y) \neq \emptyset$, 
then one of the matching segments is nested inside the other one.

\smallskip

\noindent {\bf Remark}. Since the matching segment of an unmatched point $x$ is $(x,\infty)$, 
there is no matching segment of a matched point in a nested matching which contains an unmatched point.

\smallskip
\noindent From the discussion in the previous section it is easy to see that any queue matching is nested.  The following lemma proves 
that the stable matching $\cM_s$ is also nested.

\begin{lemma}\label{stableisnested}
The stable matching $\cM_s$ is nested.
\end{lemma}
 \begin{proof}
See appendix \ref{stableisnestedproof}. 
\end{proof}

\noindent Let $\cA$ be the set of all nested matchings between points in 
$\blues$ and $\reds$. We say a red point $r$ is a {\it potential match} for a blue point $b$, if there
exists a nested matching in which $b$ is matched to $r$.
For any blue point $b\in \blues$ define $\cP(b)$ to be the set of all potential matches of $b$,
\begin{align}
\cP(b) = \{r\in \reds: \exists \cM\in \cA \text{~s.t.~} \cM(b) = r\}. \nonumber
\end{align}

\noindent The following lemma shows that the set of potential matches of any two blue points are either 
disjoint or the same.

\begin{lemma} \label{potentialmatches}
For any  $b_1, b_2 \in \blues$, $\cP(b_1) \cap \cP(b_2) \neq \emptyset$ implies  
$\cP(b_1) = \cP(b_2)$.
\end{lemma}

 \begin{proof}
See appendix \ref{continuousproofs}. 
\end{proof}

\noindent Now define the relation $\sim$ on points in $\blues$ as follow:
\begin{align}
b_1 \sim b_2 \Longleftrightarrow \cP(b_1) \cap \cP(b_2) \neq \emptyset.
\nonumber
\end{align}
According to Lemma \ref{potentialmatches}, for any $b_1, b_2 \in \blues$, if $\cP(b_1)$ and 
$\cP(b_2)$ are not disjoint, then they are the same. Therefore, $\sim$ is an equivalence relation on $\blues$. 
For any blue point $b \in \blues$, define $[b]$ to be $b$'s equivalence class in $\blues$, {\it i.e.}, 
\begin{align}
[b] \triangleq \{b' \in \blues: \cP(b') =  \cP(b)\}.
\nonumber
\end{align} 
Let $N(b)$ represent the size of $b's$ equivalence class, {\it i.e.}, $N(b) = |[b] |$.
Also define $N^+(b) = |\{b'\in [b], b'\ge b\}|$ and $N^-(b) = |\{b'\in [b], b'\le b\}|$. It is clear
that $N(b) = N^+(b) + N^-(b) - 1$.
In the following lemma we prove some facts about the structure of the equivalence classes. 

\begin{lemma}\label{equivalenceClass}
Suppose $\lambda < \mu$ and let $\blues$ and $\reds$ represent the set of blue and red points in a Poisson matching problem $PM(\lambda, \mu)$, respectively. For any given blue point $b \in \blues$
we have
\begin{enumerate}[label=(\roman*)]
\item there exist $r_1, r_2 \in \cP(b)$ such that $r_1 > b$ and $r_2 < b$, almost surely;
\item there exists exactly one potential red point $r \in \cP(b)$ between every two consecutive 
blue points in $[b]$;
\item $|\cP(b)| = N(b)+1$ on $\{|\cP(b)| < \infty\}$; and 
\item $N^+(b)$ and $N^-(b)$ are independent geometric random variables with parameter $1-\lambda/\mu$.
\end{enumerate}
\end{lemma}

\begin{proof}
See appendix \ref{third}.
\end{proof}

\noindent Let $b \in \blues$ be an arbitrary blue point. Since $\lambda<\mu$, then almost surely $|\cP(b)|<\infty$. 
From Lemma \ref{equivalenceClass}, we can conclude that blue and red points in $[b] \cup \cP(b)$ 
form a finite sequence $\{w_i\}$, for $-2N^-(b) +1 \le i \le 2N^+(b) -1$, where 
$b_0 = b$, $[b] = \{w_i: i \text{~is even} \}$, and $\cP(b)= \{w_i: i \text{~is odd} \}$. 
In other words, this sequence starts with 
a potential red point, alternates between points in $[b]$ and $\cP(b)$, and ends with 
another potential red point. We call the sequence $\{w_i\}$ 
$b$'s \emph{potential wave} and denote it by $\mathcal{W}(b)$. 
Figure $2$ shows potential waves of an instance of $PM(\lambda, \mu)$.
\footnote{This instance is the same as the instance in Figure 1. As we can see in all the 
matchings shown in Figure $1$, each point is matched to a point within the same potential wave 
shown in Figure $2$.} 
A key observation here is that in any nested matching, any blue point in $[b]$ should be matched to a red point in $\cP(b)$. Therefore, 
a nested matching first partitions $\cS$ into potential waves and then matches points 
within each wave, separately. In the following section we present our results on the analysis of the matching distance $X$ in the stable
matching $\cM_s$.

\begin{figure}[t]
\begin{center}
\includegraphics[width=.45\textwidth]{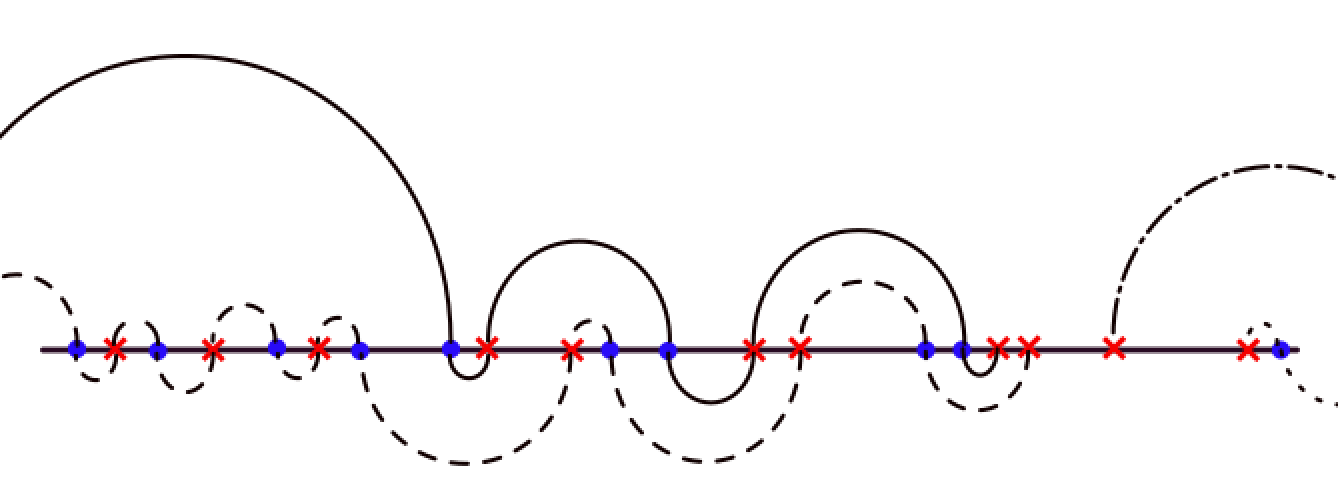}
\end{center}
\caption{Different potential waves of a sample problem.}
\label{f-nopsfrag}
\end{figure}

\subsection{Matching distance, $X$}
The following theorem, proves bounds on the distribution of $X$, in terms of busy cycles.

\begin{theorem} \label{bound}
Consider an instance of a Poisson matching problem $PM(\lambda, \mu)$, where $\lambda < \mu$.
Then we have
\begin{align}
\min\{B_1, B_1'\} < X < \max \left\{ \sum_{i=1}^{2N-1} B_i , \sum_{i=1}^{2N'-1} B_i' \right\}
\nonumber
\end{align}
where $B_i$ and $B_i'$ are i.i.d.\ random variables with distribution $B(\lambda, \mu)$, and 
$N$ and $N'$ are i.i.d.\ geometric random variables with parameter $1-{\lambda}/{\mu}$. 
\end{theorem} 
\begin{proof}
See appendix \ref{boundproof}.
\end{proof}

\noindent Using Theorem \ref{bound} we can find the following upper bound for the expected matching 
distance $\Expect(X)$.
\begin{corollary} \label{upperbound}
\begin{align}
\mathbb{E}(X) \le \left(1 + \frac{\mu+\lambda}{\mu-\lambda}\right)\frac{1}{\mu - \lambda}.
\nonumber
\end{align}
\end{corollary}
\begin{proof}
See appendix \ref{upperboundproof}
\end{proof}

\noindent \textbf{Remark}. Note that the results of Theorem \ref{bound} and Corollary \ref{upperbound} also hold if $X$ is the matching 
distance in any nested matching. 

\smallskip
\noindent In the next theorem we improve the upper bound given in Corollary \ref{upperbound} for the
expected matching distance $\Expect(X)$. The proof of this theorem is extensive and
requires some detailed combinatorial arguments. For more details see appendix \ref{forth}.

\begin{theorem} \label{Xupperbound}
For the stable matching $\cM_s$, we have
\begin{align}
\mathbb{E}(X) \le \left(1 + \ln\bigl(\frac{\mu+\lambda}{\mu-\lambda}\bigr)\right)\frac{1}{\mu - \lambda}.
\nonumber
\end{align}
\end{theorem}
\begin{proof}
See appendix \ref{forth}.
\end{proof}

In order to evaluate the goodness of the bound in Theorem \ref{Xupperbound}, note that for large values of $\mu$ ($\mu \gg \lambda$), with a high probability, each blue point will be matched to the closest red point to it. Therefore,
as $\mu/\lambda \to \infty$, $X$ converges to an exponential distribution with rate $2\mu$ (minimum of two i.i.d.\ exponentials with rate $\mu$) and $\Expect(X) \sim \frac{1}{2\mu}$. However, 
from Theorem \ref{Xupperbound}, in the limit as $\mu/\lambda \to \infty$, $\Expect(X)$
is upper bounded by $ \frac{1}{\mu}$.

\section{Conclusion}
\label{conclusion}
This paper introduced a model for studying matching markets in which preference 
lists are drawn according to distances in appropriate metric spaces, either between 
the profiles of participants or between the participants themselves.  The model
naturally captures several aspects of real world matching markets.  Various results
regarding the uniqueness and quality of stable matchings were obtained.  Specifically,
for matchings on the hypercube under the Hamming and Weighted Hamming distances, 
lower and upper bounds were obtained on the dimension of the hypercube (equal to 
the number of questions a participant in a dating site needs to answer) so as to
obtain unique stable partners or stable matchings.  Furthermore, bounds on the
distribution and the average value of the matching distance of a typical participant
(a measure of the quality of the stable matching) were obtained for stable matchings
on the hypercube and on the real line.  

We view this work as a first step in studying matching markets in the metric space
setting.  Several obvious next steps suggest themselves, notably studying the problem
under dynamic inputs; i.e., as participants arrive and depart.   
     
\bibliographystyle{ACM-Reference-Format}

\appendix
\section{Appendix}
\subsection{\textbf{Proofs omitted from section \ref{hypercubesection}}} \label{hypercubeappendix}
\textbf{Proof of Lemma \ref{sameDistanceLemma}:}
Assume, by contradiction, that there exist stable matchings $\mu_1$ and $\mu_2$ so that for
some $z\in\cS$, $r_1 = d_{\mu_1}(z) > d_{\mu_2}(z) = r_2$.  Let $r = (r_1+r_2)/2$.  
For each $x\in\cS$, let $P'_x = \{y\in P_x: d(x,y)\leq r\}$ be the preference list $P_x$ 
truncated to contain only those participants who are at a distance less or equal to $r$ 
from $x$.  Let the ordering in the truncated list $P'_x$ be denoted by $\succ_x'$.  
Call the PMP restricted to the truncated preference lists as the ``truncated matching problem''.
In the truncated matching problem, each person prefers to remain unmatched than to match with 
a person at a distance greater than $r$ from them.  Let $\mu$ be a stable matching for the PMP
which has stable partners with a matching distance greater than $r$.  Construct the partial 
matching $\mu'$ from $\mu$ by removing all pairs with a distance greater than $r$. 
We show that $\mu'$ is a stable matching for the truncated matching problem.  Suppose $m$ and $w$ 
are not matched to each other in $\mu'$.  If $d(m,w) > r$, then clearly $(m,w)$ cannot 
form a blocking pair for the truncated matching problem.  Suppose $d(m,w) \le r$.
Since $m$ and $w$ are not matched to each other in $\mu'$, they cannot be matched to 
each other in $\mu$.  Moreover, since $\mu$ is stable, either $\mu(m) \succ_m w$, or 
$\mu(w) \succ_w m$.  Without loss of generality, assume $\mu(m) \succ_m w$. Therefore, 
\begin{align}
d\bigl(m,\mu(m)\bigr) \le d(m,w)\le r \Rightarrow \mu(m) \succ_m' w.
\nonumber
\end{align}
This implies $m$ is also matched to $\mu(m)$ in $\mu'$; \emph{i.e.}, $\mu'(m) = \mu(m)$. 
Therefore, $\mu'(m) \succ_m' w$ and $(m,w)$ cannot be a blocking pair for $\mu'$. 
This proves that $\mu'$ is a stable matching for the truncated matching problem.
Now define $S_r(\mu)$ to be the set of all users who are matched to someone at a 
distance greater than $r$ in $\mu$,
$$S_r(\mu) \overset{\Delta}{=} \left\{x\in S: d_\mu(x) > r\right\}.$$
It is clear that $S_r(\mu)$ is the same set of users who are not matched in $\mu'$.
By the Rural Hospital Theorem, the set of unmatched men and women in the 
truncated matching problem is the same in all stable matchings. 
This implies $\cS_r(\mu)$ does not depend on $\mu$.
This contradicts our initial assumption, since $z\in\cS_r(\mu_2)$ but $z\notin\cS_r(\mu_1)$, 
proving the lemma.
\epr

Consider the RPMP-$k$.  Let $\cM$ and $\cW$ represent the set of $n$ men and $n$ women, respectively. 
For any profile $\ba = (a_1, ..., a_k) \in Q_k$, let $\cM_\ba$ and $\cW_\ba$ be the sets of all men 
and women whose profiles equal $\ba$, respectively.  Define $\cS_\ba = \cM_\ba \cup \cW_\ba$. 

\begin{lemma}\label{claim}
Fix $\ba \in Q_k$ and without loss of generality assume $|\cM_\ba| \le |\cW_\ba|$.  We claim 
that in every stable matching, each man in $\cM_\ba$ will be matched to a woman in $\cW_\ba$. 
\end{lemma}
\begin{proof}
Suppose to the contrary that there is a stable matching $\mu$ and an $m\in \cM_\ba$ such that 
$\mu(m) \notin \cW_\ba$.  Since $|\cM_\ba| \le |\cW_\ba|$, there should also exist a woman 
$w \in \cW_\ba$ such that $\mu(w) \notin \cM_\ba$.  However, since $d(m, w) = 0$, $w\succ_m \mu(m)$ 
and $m \succ_w \mu(w)$.  Therefore, $(m,w)$ forms a blocking pair for $\mu$, which is a contradiction.
\end{proof}

Thus, for any $\ba \in Q_k$, every stable matching should first try to match men in $\cM_\ba$ with
women in $\cW_\ba$ according to their tie-breaking preference lists.  Any one unmatched woman in 
$|\cW_\ba|$ will be matched to someone at a further distance.  Define $\cU_\ba$ to be the set of 
all users with profile $\ba$, which are matched to someone with a  profile different from $\ba$.  
Note that according to the Rural Hospital Theorem, $\cU_\ba$ is the same for all stable matchings 
and $|\cU_\ba| = \big{|} |\cM_\ba| - |\cW_\ba|\big{|}$.  Define $\cS_\ba = \cM_\ba \cup \cW_\ba$. 
The following lemma shows that if $k$ is constant, then for any $\ba \in Q_k$, $|\cU_\ba| = O(\sqrt{n})$.
\begin{lemma}\label{uasize}
For any arbitrary profile $\ba \in Q_k$, as $n \to \infty$,
\begin{align}
\frac{|\cM_\ba| - |\cW_\ba|}{\sqrt{2p(1-p)n}} \xrightarrow{~d~} Z_1 \quad \text{and} \quad
\frac{|\cS_\ba| - 2np}{\sqrt{2p(1-p)n}} \xrightarrow{~d~} Z_2, \nonumber
\end{align}
where $p = 2^{-k}$, $Z_1$ and $Z_2$ are independent standard normal--$\cN(0,1)$--random variables,
and $\xrightarrow{~d~}$ represents convergence in distribution. 
\end{lemma}

\begin{proof}
First note that $|\cM_\ba|$ and $|\cW_\ba|$ are i.i.d.\ with a Binomial$(n,p)$ distribution. From Central Limit Theorem (CLT) we have that as $n \to \infty$,
\begin{align}
\frac{|\cM_\ba| - np}{\sqrt{np(1-p)}} \xrightarrow{~d~} {N}_1, \quad
\frac{|\cW_\ba| - np}{\sqrt{np(1-p)}} \xrightarrow{~d~} {N}_2,\nonumber
\end{align}
where $N_1$ and $N_2$ are two independent random variables with a standard normal
distribution, {\it i.e.}, $N_1, N_2 \sim \mathcal{N}(0,1)$. Therefore, as $n \to \infty$,
\begin{align}
\frac{|\cM_\ba| - |\cW_\ba|}{\sqrt{2p(1-p)n}} \xrightarrow{~d~} \frac{N_1 - N_2}{\sqrt{2}},  \quad
\frac{|\cS_\ba| - 2np}{\sqrt{2p(1-p)n}} \xrightarrow{~d~}  \frac{N_1 + N_2}{\sqrt{2}}. \nonumber
\end{align}
Define $Z_1=({N_1 - N_2})/{\sqrt{2}}$ and $Z_2 = (N_1 + N_2)/{\sqrt{2}}$. It is clear that 
$Z_1, Z_2 \sim \cN(0,1)$. Moreover, since $N_1$ and $N_2$ are independent, $Z_1$ and 
$Z_2$ are also independent. This completes the proof. 
\end{proof}

Since Lemma \ref{claim} requires each stable matching $\mu$ to first match men and women in 
$\cS_\ba$ using their tie-breaking preference lists, the $O(\sqrt{n})$ discrepancy between the number of 
men and women in $\cS_\ba$ makes this sub-problem significantly unbalanced.  Using the approach
of \citet{ashlagi2014unbalanced}, we prove some useful bounds on the number of 
stable partners in unbalanced matching problems which is true for every $n$.
\begin{lemma} \label{unbalanced}
Let $r \ge 1$ and consider an unbalanced two-sided matching problem with $n$ men and $n+r$ women 
represented by $\cM$ and $\cW$, respectively.
Suppose the men's preference lists are generated independently and uniformly at random from the 
set of all orderings of women in $\cW$. Similarly, suppose the women's preference lists are generated 
independently and uniformly at random from the set of all orderings of men in $\cM$. For any given
 $x \in \cM \cup \cW$, let $N(x)$ represent the number of $x$'s stable partners. We have that
\begin{align}
 \Prob\big{(} N(x) > 1 \big{)} \le \frac{1}{r+1}  \quad \text{and} \quad
 \Expect\big{(} N(x) \big{)} \le 1+ \frac{1}{r}. \nonumber
\end{align}
\end{lemma}

\begin{proof}
Let $\mu_\cM$ represent the men-optimal stable matching found by running the 
men-proposing deferred acceptance algorithm, and let $U$ be the set of all women 
who are not matched in $\mu_\cM$. 
According to the Rural Hospital Theorem, the set of women who are unmatched is the same
as $U$ for all stable matchings. Let $w \in \cW \backslash U$ be an arbitrary woman. 
In order to find all the stable partners of $w$, 
we employ the same algorithm that is used in \citet{mcvitie1970stable}, 
\citet{immorlica2005marriage}, and \citet{ashlagi2014unbalanced}. 
It has been proved by \citet{immorlica2005marriage} that the following algorithm outputs
all the stable partners of $w$.
\\
\\
\emph{Algorithm I}
\begin{enumerate}
\item Run the men-proposing algorithm to find the men-optimal stable matching $\mu_\cM$. If
	$w$ is unmatched in $\mu_\cM$, output $\emptyset$. Initialze $\mu = \mu_\cM$ .
\item Set $m=\mu(w)$ and output $m$ as one of the stable partners of $w$. Then have $w$ reject $m$ and remove the pair $(m,w)$ from $\mu$. Set $u=m$. ($u$ represents the current unmatched
man.)
\item Let $w'$ be the next woman in $u$'s preference list whom he has not proposed to yet. If 
$w'$ is unmatched in $\mu_\cM$, terminate the algorithm.
\item
\begin{enumerate}
\item If $w'$ has already received a proposal from someone better than, she simply 
rejects $u$ and the algorithm continues to step $3$.
\item If not, $w'$ accepts $u$'s proposal. If $w = w'$, the algorithm continues to
step $2$. Otherwise, set $u = \mu(m')$ and the algorithm continues to step $3$.
\end{enumerate} 
\end{enumerate}
\vspace{8pt}
In order to analyze algorithm I, we use the principle of deferred decision which assumes 
that the random preference lists are not known in advance and rather unfold step by step
in the algorithms when a proposal/rejection happens.
Let $t_i$ be the time of the $i^{\text th}$ visit of the algorithm at step $3$, and define $u_i$
and $w_i'$ to be the unmatched man and the next woman who $u_i$ wants to propose to at 
time $t_i$. Also define $X_i$ to be the set of all women who $u_i$ has not proposed to yet at 
time $t_i$.
Since we are using the principle of deferred decision, at any time $t_i$, 
rankings of women in $X_i$ are not yet unfolded in $u_i$'s preference list. 
Therefore, at any time $t_i$, every woman in $X_i$ has the same chance of 
${1}/{|X_i|}$ to receive the next proposal from $u_i$. Define the events 
$E_i = \left\{w_i' \in \{w\} \cup U\right\}$. 
Since the algorithm has not been terminated by time $t_i$, $U \subseteq X_i$. Therefore,
given $E_i$, the probability that $u_i$ proposes to $w$ is at most $1/(r+1)$, and the probability
that the algorithm terminates is at least $r/(r+1)$, {\it i.e.}\ ~,
\begin{align} 
\Prob(w_i' = w ~ | ~ E_i) \le \frac{1}{r+1},  \nonumber
\end{align}
and
\begin{align}
\Prob(\text{The algorithm terminates at~} t_i ~ | ~ E_i) \ge \frac{r}{r+1}. \nonumber 
\end{align}
As the algorithm progresses, woman $w$ finds a new stable partner only if she receives a 
proposal from an unmatched man at step $3$ of the algorithm. Let $V_i$ be the total number 
of proposals received by woman $w$ from time $t_1$ to time $t_i$. 
If $E_i$ does not occur then $V_{i+1} = V_i$, and if $E_i$ occurs then $V_{i+1} = V_i+1$
with a probability of at most $1/(r+1)$ and the algorithm terminates with a probability at least 
$r/(r+1)$. Therefore, if $V$ represents the total number of proposals received by $w$ after 
time $t_1$, V is stochastically dominated by a geometric random variable 
with rate $p = r/(r+1)$. Thus,
\begin{align}
\Prob(V > 0) \le \frac{1}{r+1},  \nonumber
\end{align}
 and
 \begin{align}
 \Expect(V) \le \frac{1-p}{p} = \frac{1}{r}. \nonumber
 \end{align}
Since $\cN(w) \le 1 + V$, the proof is complete for any $w\in \cW$. It remains to prove the 
inequalities for $x \in \cM$. Fix $x\in \cM$. Note that the two events $\{\cN(x) > 1\}$ and 
$\{\cN(\mu_\cM(x)) > 1\}$ are equivalent. Therefore, Since $\mu_\cM(x) \in \cW$,
\begin{align}
\Prob\big{(} \cN(x) > 1 \big{)} =  \Prob\big{(} \cN(\mu_\cM(x)) > 1 \big{)} \le \frac{1}{r+1}. 
\nonumber
\end{align}
Moreover, since $\sum_{m \in \cM} \cN(m) = \sum_{w\in \cW\backslash U} \cN(w)$ (both are 
equal to the total number of stable partner pairs), from symmetry we have,

\begin{align}
\Expect\left(N(x)\right) = \frac{1}{n} \sum_{m \in \cM} \Expect\left(\cN(m)\right) = 
\frac{1}{n} \sum_{w\in \cW\backslash U} \Expect\left(\cN(w) \right) \le 1 + \frac{1}{r}. \nonumber
\end{align}
\end{proof}


\noindent
We now prove Theorem \ref{smallk} by using Lemmas \ref{uasize} and \ref{unbalanced}.

\smallskip

\noindent \textbf{Proof of Theorem \ref{smallk}:}
\textbf{Part (i)}.
We prove this part of the theorem only for constant profile size $k\ge 1$. The proof for arbitrary 
profile size $k \le (1-\epsilon) \log n$ is similar. Fix $n$ and consider an instance of the random profile matching 
problem with $n$ men and $n$ women represented by $\cM^{(n)}$ and $\cW^{(n)}$, respectively.
Let $x \in \cM^{(n)} \cup \cW^{(n)}$ be an arbitrary user and let $\ba\in Q_k$ represent his/her profile. Define 
$\cM_\ba^{(n)}$, $\cW_\ba^{(n)}$, and $\cS_\ba^{(n)}$ as before. Also define
$Z_1^{(n)} = {(|\cS_\ba^{(n)}| - 2np)}/{\sqrt{2np(1-p)}}$ and 
$Z_2^{(n)} = {(|\cW_\ba^{(n)}| - |\cM_\ba^{(n)}|)}/{\sqrt{2np(1-p)}}$. 
According to Lemma \ref{uasize}, as $n \to \infty$, $Z_1^{(n)} \xrightarrow{~d~} Z_1$ and 
$Z_2^{(n)} \xrightarrow{~d~} Z_2$, where $Z_1, Z_2 \sim \cN(0,1)$, and $Z_1$ and $Z_2$ are independent.
For any $\epsilon>0$ define subsets $A_\epsilon, B_\epsilon \subseteq \reals$ as follows,
\begin{align}
A_\epsilon =  [-\infty, -\frac{1}{\epsilon}] \cup [\frac{1}{\epsilon}, +\infty], \quad
B_\epsilon =[0, \epsilon] \cup [\frac{1}{\epsilon}, +\infty]. \nonumber
\end{align}
Let $\delta >0$ be an arbitrary positive number. Choose $\epsilon >0$ small enough to have,
${\Prob}\left( Z_1\in A_\epsilon\right) \le \frac{\delta}{8}$ and 
${\Prob}\left(  |Z_2|\in B_\epsilon\right) \le \frac{\delta}{8}$. 
Since $Z_1^{(n)}$ and $Z_2^{(n)}$ converge in distribution to $Z_1$ and $Z_2$, respectively, there exists a large number $N_1$ such that for any $n > N_1$,
\begin{align}
\bigl|\Prob( Z_1\in A_\epsilon) - \Prob( Z_1^{(n)}\in A_\epsilon)\bigr| < \frac{\delta}{8},
\quad
\bigl|\Prob( Z_2\in B_\epsilon) - \Prob( Z_2^{(n)}\in B_\epsilon)\bigr| < \frac{\delta}{8}.
\nonumber
\end{align}
Therefore, for any $n>N_1$ we have,
\begin{align}
\Prob\left( Z_1^{(n)} \in A_\epsilon \text{~~or~~} Z_2^{(n)} \in B_\epsilon \right) & \le \nonumber
\Prob( Z_1^{(n)} \in A_\epsilon ) + \mathbb{P}(Z_2^{(n)} \in B_\epsilon )\\ \nonumber
& \le {\Prob}( Z_1 \in A_\epsilon) + \bigl|\Prob( Z_1^{(n)}\in A_\epsilon) - \Prob( Z_1\in A_\epsilon)\bigr|\\
\nonumber
& + {\Prob}( Z_2 \in B_\epsilon) + \bigl|\Prob( Z_2^{(n)}\in B_\epsilon) - \Prob( Z_2\in B_\epsilon)\bigr|\\ 
\nonumber
& \le \frac{\delta}{2}. \nonumber
\end{align}
Therefore, with a probability of at least $1-{\delta}/{2}$, the following event occurs:
\begin{align}
E = \Big{\{}|\cS_\ba^{(n)}| \in [2np-C_1\sqrt{n}, 2np+C_1\sqrt{n}] \text{~~and~~} 
\big{|}|\cW_\ba^{(n)}| - |\cM_\ba^{(n)}|\big{|} 
\in [C_2\sqrt{n}, C_3\sqrt{n}] \Big{\}}, \nonumber
\end{align} 
where
$C_1 = \frac{\sqrt{2p(1-p)}}{\epsilon}$, $C_2 = \epsilon\sqrt{p(1-p)}$, and $C_3 = \frac{\sqrt{p(1-p)}}{\epsilon}$. 
Without loss of generality assume $|\cM_\ba^{(n)}| \le |\cW_\ba^{(n)}|$ and define 
$r = |\cW_\ba^{(n)}| - |\cM_\ba^{(n)}|$. 
The problem of matching men in $\cM_\ba^{(n)}$ and women in $\cW_\ba^{(n)}$ according to 
preference lists given by $\succ^P$ is an unbalanced matching problem with a discrepancy 
equal to $r$ between the number of men and the number of women. 
Let $\cU_a^{(n)}$ represent the set of unmatched women in the men-optimal stable matching 
for this unbalanced matching problem. Therefore, if $N^{(n)}(x)$ represents the number of 
$x$'s stable partners,  we have,
\begin{align}
{\Prob}\big{(} N^{(n)}(x) > 1 \big{)} 
& \le {\Prob}{(} E^o {)} + {\Prob}\big{(} N^{(n)}(x) > 1 \big{|} E\big{) } \nonumber \\ \nonumber
  & = {\Prob}{(} E^o {)} + 
  {\Prob}\big{(} x \in \cU_a^{(n)} \big{|} E\big{)} + 
  {\Prob}\big{(} N^{(n)}(x) > 1 \big{|} E, x \notin \cU_a^{(n)} \big{)}\\ \nonumber
  & \le {\Prob}{(} E^o {)} + \frac{r}{|\cS_\ba^{(n)}|} + \frac{1}{r+1}\\ \nonumber
  & \le \frac{\delta}{2} + \frac{C_2\sqrt{n}}{2np - C_1\sqrt{n}} + \frac{1}{C_1\sqrt{n}+1}, \nonumber
\end{align}
 where in the first inequality we used the results of the Lemma \ref{unbalanced}, 
 and in the last inequality we used the bounds on $r$ and $|\cS_\ba^{(n)}|$ given by
 the event $E$. Pick $N_2$ large enough to have,
 \begin{align}
 \frac{C_2\sqrt{n}}{2np - C_1\sqrt{n}} + \frac{1}{C_1\sqrt{n}+1} < {\delta}/{2}, \quad \forall n \ge N_2. 
 \nonumber
 \end{align}
 Define $N = \max\{N_1,N_2\}$. Therefore, for any $n>N$, 
 $ {\Prob}\big{(} N^{(n)}(x) >1 \big{)} \le \delta$.
\noindent Since $\delta >0$ is arbitrary,
\begin{align}
\Prob\big{(} N^{(n)}(x) =1 \big{)} \to 1 \text{~~as~~ } n\to \infty. \nonumber
\end{align}
This implies that with high probability the fraction of users with multiple stable partners
tends to zero as $n \to \infty$. \\
\textbf{Part (ii)}. 
In order to prove the second part of the theorem, 
note that since $|\cM_\ba|$ and $|\cW_\ba|$ are Binomial random variables 
with parameters $n$ and $p = 2^{-k} = 1/n$, according to the well-known Poisson 
limit theorem, both converge to the Poisson$(1)$ distribution as $n$ goes to infinity. 
Therefore, in the limit, with a positive probability of $c = {e^{-2}}/{4}$ there
are exactly two men and two women whose profiles are equal to $\ba$. 
On the other hand, it is easy to see that in a random stable matching problem with two 
men and two women, the probability of having exactly two stable matchings 
is equal to $1/8$.
 Therefore, for any given profile $\ba \in Q_k$, with a positive probability of $\beta = {c}/{8} >0$, 
there are exactly two men and two women with profile $\ba$ who have multiple 
stable partners. This proves that the expected number of users with multiple stable partners
is $O(n)$. Moreover, since the number of such profiles is $O(n)$, in expectation there are 
exponentially many stable matchings.
\epr
The following lemma shows that if the preference list of a user is uniquely identified by profile 
distances and no further tie-breaking is required, then he/she has a unique stable partner. 
\begin{lemma} \label{uniqueStablePartner2}
In a profile matching problem, if the distances of a given user $x$ from all the members 
of the opposite sex are distinct, then $x$ has a unique stable partner.
\end{lemma}

\begin{proof}
By contradiction, suppose $x$ has two different stable partners $y_1$ and $y_2$. According
to Lemma \ref{sameDistanceLemma}, $y_1$ and $y_2$ should be at the same distance from 
$x$. But, this contradicts with the assumption that $x$ has different distances from $y_1$ and $y_2$. 
Therefore, $x$ has a unique stable partner.
\end{proof}
In order to apply Lemma \ref{uniqueStablePartner2}, $k$ should be large enough 
to have a unique stable matching without resorting to tie-breaks. 
Now we prove Theorems \ref{hamming} and \ref{weightedHamming}.

\smallskip

\noindent \textbf{Proof of Theorem \ref{hamming}:}
\textbf{Part (i)}.
Let $x$ be an arbitrary user and without loss of generality, assume $x\in \cW$.
Suppose $x$ has multiple stable partners and let $y$ and $y'$ be two different stable partners of $x$. 
Since $x$ has multiple sable partners, $y$ also has another stable partner $x'$ (different from $x$). 
According to Lemma \ref{sameDistanceLemma}, $d_h(x,y') = d_h(x,y) = d_h(x',y)$. For any $z\in \cM$ define the following event 
$$E_z = \bigl\{\exists x'\in \cW\backslash \{x\}, \exists z'\in \cM\backslash \{z\} 
;d_h(x,z') = d_h(x,z) = d_h(x',z)\bigr\}.$$
Using the union bound we have,
\begin{align}
\Prob(E_z) &= \Prob\Bigl( \exists x'\in \cW\backslash \{x\}; d_h(x,z) = d_h(x',z)\Bigr)
\Prob\Bigl( \exists z'\in \cM \backslash \{z\}; d_h(x,z) = d_h(x,z') \Bigr) \nonumber \\
&\le\Bigl((n-1) \Prob\bigl(d_h(x,m) = d_h(x,z)\bigr) \Bigr)
\Bigl((n-1) \Prob\bigl(d_h(w,z) = d_h(x,z)\bigr) \Bigr)  \nonumber \\
&\le n^2 \Prob\bigl(d_h(x,m) = d_h(x,z)\bigr)^2 \nonumber
\nonumber
\end{align}
where $m$ and $w$ are a man and a woman who are chosen randomly from $\cM$ and $\cW$, 
respectively.
Note that in the last inequality we used the existing symmetry in the problem. Since 
$d_h(x,m)$ has a binomial distribution (as a function of the random variable $m$), 
the maximum value of $\Prob\bigl(d_h(x,m)=d_h(x,z)\bigr)$ is at $d_h(x,z) = \lfloor k/2 \rfloor$. Using the 
Sterling approximation we have:
\begin{align}
{\Prob}\bigl( d_h(x,m)= \frac{k}{2}\bigr) = \frac{k!}{ \frac{k}{2}!\frac{k}{2}! } 2^{-k} \simeq 
\frac{\sqrt{2\pi k}(\frac{k}{e})^k}{\pi k(\frac{k}{2e})^k} 2^{-k} = \sqrt{\frac{2}{\pi}} \frac{1}{\sqrt{k}}.
\nonumber
\end{align}
Therefore,
\begin{align}
\Prob(x \text{ has multiple stable partners}) \le \Prob\left(\cup_{z\in \cM} E_z\right) 
 \le n^3 {\frac{2}{\pi k}}. \nonumber 
\end{align}
Since $k = \Omega(n^3)$, the right hand side of the above inequality tends to zero as $n$ goes
to infinity. This implies that with high probability the fraction of users with multiple stable partners tends to
zero as $n$ goes to infinity. Note that using the union bound, we can conclude that if $k = \Omega(n^4)$,
with high probability, there exists a unique stable matching.\\
\textbf{Part (ii)}. 
Fix a woman $x \in \cW$. For any $y\in \cM$, define the event $E_y = \{\exists y' \in \cM \backslash \{y\};
d_h(x,y) = d_h(x,y')\}$. Also define $A_x$ to represent the event that the distances of $x$ from all men in $
\cM$ are distinct. Similar to part (i) we have
$$\Prob(A_x^{c}) \le \Prob\left(\cup_{y\in \cM} E_y\right) \le n^2\sqrt{\frac{2}{k\pi}}.$$
According to Lemma \ref{uniqueStablePartner2}, if the event $A_x$ happens for every $x\in \cW$, 
there is a unique stable matching without resorting to tie-breaking. Therefore,
\begin{align}
\Prob (\text{There are multiple stable matchings}) &\le 
\Prob \left( \cup_{x\in\cW} A_x^{c} \right) \nonumber 
\le n^3\sqrt{\frac{2}{k\pi}}. \nonumber
\end{align} 
Since $k = \Omega(n^6)$, the probability that there are multiple stable matchings 
goes to zero as $n$ goes to infinity.\vspace{8pt} \\
\textbf{Part (iii)}.
Without loss of generality assume $x\in \cW$ and let $X_i$ represent the distance of $x$ from man $m_i$, {\it i.e.}, $d_i = d_h(x,m_i)$. Clearly, $X_i$'s are i.i.d.\ with Binomial distribution with 
parameters $k$ and $1/2$. Define $Z = \min_{i}X_i$. Clearly $d_h(x) \ge Z$. Therefore, for any 
positive number $r>0$,
\begin{align}
\Prob(d_h(x)\ge r) \ge \Prob(Z\ge r) = \Prob(\min_i X_i\ge r) = \prod_{i=1}^n \Prob(X_i\ge r) = \Prob(X_1\ge r)^n.
\nonumber
\end{align}
Now, according to the Chernoff's inequality,
\begin{align}
\Prob(X_1<r) \le e^{-\frac{1}{{k}}{(\frac{k}{2}-r)^2}}.
\nonumber
\end{align}
Therefore,
\begin{align}
\Prob(d_h(x)\ge r) \ge (1 -e^{-\frac{1}{{k}}{(\frac{k}{2}-r)^2}} )^n \ge 1-ne^{-\frac{1}{{k}}{(\frac{k}{2}-r)^2}}.
\nonumber
\end{align}
Now if we set $r = k/2 - \sqrt{\beta k\log n}$ we have,
\begin{align}
\Prob\left(d_h(x)\ge k/2 - \sqrt{\beta k\log n}\right) \ge 1-n^{1-\beta}.
\nonumber
\end{align} 
\epr

\smallskip

\noindent \textbf{Proof of Theorem \ref{weightedHamming}:}
\textbf{Part (i)}.
Without loss of generality assume $x \in \cM$ and suppose $y$ is a stable partner for $x$.
If $x$ has multiple stable partners, then $y$ should also have multiple stable partners. Let
$x'$ be another stable partner of $y$ different from $x$. According to Lemma 
\ref{sameDistanceLemma}, $x$ and $x'$ should have the same distance from $y$. 
However, in the weighted hamming distance metric, if
$d_w(\ba(y), \ba(x)) = d_w(\ba(y), \ba(x'))$, then $x$ and $x'$ should have 
the exact same profiles, {\it i.e.}, $\ba(x) = \ba(x')$. Therefore, if a man $x$ has 
multiple stable partners, there should exist another man $x'$ with the same profile 
as him. However, by using union bounds we get
\begin{align}
\Prob\left( \exists x' \in \cM; \ba(x') = \ba(x) \right) \le (n-1) \frac{1}{2^k} \le n^{-\epsilon}. \nonumber
\end{align}
This proves that the probability that $x$ has multiple stable partners is 
vanishingly small.\vspace{8pt} \\
\textbf{Part (ii)}.
We first show that if there are multiple stable matchings, then there are two men (or women) who
have the same profile. If there are multiple stable matchings, there should exists a chain 
$\{c_i\}_{i = 0}^{2s-1}$ of men and women such that $c_i \in \cM$ if $i$ is even, and $c_i \in \cW$ 
if $i$ is odd. Moreover, in this chain
for every $i \in \{0, \dots, 2s-1\}$, $c_{i+1} \succ_{c_i} c_{i-1}$ ($i-1$ and $i+1$ are taken in mode $2s$).
Define $d_i = d_w(c_i, c_{i+1})$, $0\le i < 2s$. Let $j$ be the index at which $d_i$ is minimum. Now, 
if the values of $d_i$ are all distinct, then $c_{j} \succ_{c_{j+1}} c_{j+2}$ which is a contradiction. Therefore,
There should exist an index $i$ such that, $d_i = d_{i+1}$. Following our discussion in part (i), this implies
that $c_i$ and $c_{i+2}$ should have the same profile. However, the probability that two randomly 
selected men (or women) have the same profile is $2^{-k}$. Using union bound we can conclude
that the probability that there are multiple stable matchings is upper bounded by 
$n^22^{-k} \le n^{-\epsilon}$.\\
\\
\textbf{part (iii)}
In order to analyze $X$, fix $n$ and let $\epsilon >0$ be an arbitrary
positive real number. Let $\mu$ be an arbitrary stable matching. With out loss of generality assume
$x \in \cW$ and define $\ba  = \ba(x)$. 
For any positive integer $r$, define 
$\cS_\ba(r)$ as the set of all the users who have the same answers as $x$ for the first $r$ questions,
\begin{align}
\cS_\ba(r) \triangleq \left\{y\in \cM \cup \cW: \ba(y)[1:r] = \ba[1:r]\right\}. \nonumber
\end{align}
First note that for any given users $x, y \in \cS_\ba(r)$ and $y' \notin \cS_\ba(r)$,
\begin{align}
d_w(x, y)\le 2^{-r} <d_w(x,y') \Longrightarrow y\succ_x y'. \nonumber
\end{align}
Therefore, if we define $\cW_\ba(r) = \cS_\ba(r) \cap \cW$ and $\cM_\ba(r) = \cS_\ba(r) \cap \cM$, 
then the number of users in $\cS_\ba(r)$ who are not matched to someone in $\cS_\ba(r)$ in 
$\mu$ is $\big{|}|\cW_\ba(r)|-|\cM_\ba(r)|\big{|}$. 
Set $r = \lfloor (1-\epsilon) \log n\rfloor$ and define $p_r  = 2^{-r} \ge n^{ \epsilon-1}$.
Since $|\cW_\ba|$ and $|\cM_\ba|$ have Binomial distributions with parameters $n$ and $p_r$, 
following the discussions in Lemma \ref{uasize},
$|\cS_\ba(r)| = O(p_rn)$ and $\big{|}|\cW_\ba(r)|-|\cM_\ba(r)|\big{|} = O(\sqrt{p_rn})$. 
Therefore, 
\begin{align}
\Prob\left(d_w(x) > 2^{-r}\right) = \Prob\left( \mu(x) \notin \cS_\ba(r) \right) = 
\frac{\big{|}|\cW_\ba(r)|-|\cM_\ba(r)|\big{|}}{|\cS_\ba(r)|} = 
O(1/\sqrt{p_rn}) = O(n^{-\frac{\epsilon}{2}}).
\nonumber
\end{align}
This proves that,
\begin{align}
\Prob \left( \frac{\log {d_w(x)}}{\log n} > -1+\epsilon \right) \to 0, \quad \text{as} \quad n\to \infty. 
\end{align}
On the other hand, if we set $r' = \lceil (1+\epsilon) \log n \rceil$, then,
\begin{align}
\Prob\left(d_w(x) \le 2^{-r'}\right) = \Prob\left( \mu(x) \in \cS_\ba(r') \right) \le
\Prob\left( |\cM_\ba(r')| > 0\right) \le nP_{r'} = n^{-\epsilon},
\nonumber
\end{align}
where in the last inequality we used the union bound inequality. Therefore,
\begin{align}
\Prob \left( \frac{\log {d_w(x)}}{\log n} < -1-\epsilon \right) \to 0, \quad \text{as} \quad n\to \infty. 
\end{align}
From $(1)$ and $(2)$ we can conclude that ${\log {d_w(x)}}/{\log n}$ converges to $-1$ in 
probability. 
\epr

\subsection{\textbf{Proofs omitted from section \ref{continuous}}} \label{continuousproofs}

\noindent \textbf{Proof of Lemma \ref{stableisnested}:} \label{stableisnestedproof}
Let $x, y \in S$ and suppose $x \in I(y)$. Assume, by contradiction, $\cM_s(x) \notin \overline{I}(y)$. Therefore, 
 either $y \in I(x)$ or $\cM_s(y) \in I(x)$. Without loss of generality, assume $y \in I(x)$.
If $x$ and $y$ have different colors, since 
\begin{align}
|x-y| < |I(x)| = |x - \cM_s(x)|, \quad
|x-y| < |I(y)| = |y - \cM_s(y)|,
\nonumber
\end{align}
$(x,y)$ will form a blocking pair for $\cM_s$ which is a contradiction. 
Now suppose $x$ and $y$ have the same color and without loss of 
generality, assume $|\cM_s(y) - x| \le |\cM_s(x) - y|$. Since
\begin{align}
|\cM_s(y) - x| &< |I(y)| = |\cM_s(y) - y|,  \quad
|\cM_s(y) - x| &\le |\cM_s(x) - y| < |I(x)| = |x - \cM_s(x)|, \nonumber
\end{align}
$(\cM_s(y), x)$ is a blocking pair for $\cM$ which is a contradiction. 
Therefore, $\cM_s(x) \in \overline{I}(y)$ and the proof is complete.
\epr

For any $x, y \in \reals$ where $x < y$, define $g(x,y)$ to be the difference between the number of red 
and blue points in the open interval $(x,y)$,
\begin{align}
g(x,y) \triangleq \left| \{r\in \reds: r \in (x,y) \}\right| -  \left| \{b\in \blues: b \in (x,y) \}\right|.
\nonumber
\end{align}
For simplicity of notation, for $x>y$ let $g(x,y) = g(y,x)$. 
Let $\cM$ be a nested matching and let $b \in \blues$ be an arbitrary blue point which is matched 
under $\cM$. Since $\cM$ is nested, any red or blue point on $x$'s matching segment $I_\cM(x)$, 
should be matched to a point on $I_\cM(x)$. Therefore, there should be an equal 
number of red and blue points on $I_\cM(x)$. This implies that $g\bigl(x,\cM(x)\bigr) = 0$. 

\begin{lemma} \label{potentialmatches1}
For any $b\in \blues$,
\begin{align}
\cP(b) &= \{r \in \reds: g(b,r) = 0\}. \nonumber 
\end{align}
Moreover, for any  $b_1, b_2 \in \blues$, $\cP(b_1) \cap \cP(b_2) \neq \emptyset$ implies  
$\cP(b_1) = \cP(b_2)$ and $g(b_1, b_2) = 1$.
\end{lemma}

\begin{proof}
Let $Q(b)$ represent the right hand side of the equation in the lemma. 
We want to prove that $\cP(b) = Q(b)$. For any $r \in \cP(b)$, there exists a nested matching 
$\cM$ such that $\cM(b)=r$. Therefore, $g(b,r) = 0$ and this implies that $\cP(b) \subseteq Q(b)$. 
Now let $r \in Q(b)$. Let $\cS_1 \subseteq \cS$ represent the set of all the point in $\cS$ 
between $b$ and $r$, and let $\cS' = \cS \backslash \bigl\{\cS_1 \cup \{b,r\} \bigr\}$.
Since $r \in Q(b)$, $\cS_1$ has an equal number of red and blue points. Let $\cM_1$ represent 
the stable matching for points in $\cS_1$, and let $\cM_1'$ represent the stable matching
for points in $\cS_1'$. According to Lemma \ref{stableisnested}, both $\cM_1$ and $\cM_1'$ are
nested. Let $\cM$ be the matching in which $b$ is matched to $r$, and points in $\cS_1$ 
and $\cS_1'$ are matched according to matchings $\cM_1$ and $\cM_1'$, respectively.
It is clear that since $\cM_1$ and $\cM_1'$ are nested, $\cM$ is also nested. 
Therefore, $r$ is $b$'s potential match and $r \in \cP(b)$. 
This proves that $Q(b) \subseteq \cP(b)$, and the proof for the first part of the lemma 
is complete.
Now let $b_1, b_2\in \blues$ be two arbitrary blue points. Without loss of generality assume 
$b_1<b_2$. Let $r \in \reds$ be an arbitrary red
 point. There are three different possibilities for $r$'s placement with respect to $b_1$ and $b_2$:
 \begin{enumerate}
 \item $b_1<b_2<r \Rightarrow g(b_1,b_2) = g(b_1, y)-g(b_2, r) + 1$.
 \item $b_1<r<b_2 \Rightarrow g(b_1,b_2) = g(b_1, r)+g(y, b_2) + 1$.  
 \item $r<b_1<b_2 \Rightarrow g(b_1,b_2) = g(y, b_2)-g(r, b_1) + 1$.  
 \end{enumerate}
 Now if $\cP(b_1) \cap \cP(b_2) \neq \emptyset$, there exists a red point $r^*$ such that $r^* \in \cP(b_1)$ and
 $r^* \in \cP(b_2)$. Therefore, $g(b_1, r^*) = g(b_2, r^*) = 0$. According to the equations described
 above, we should have $g(b_1, b_2) = 1$ (in all there cases). Now since $g(b_1, b_2) = 1$,
 with a same argument, we can conclude that for any $r\in \reds$, 
 $g(b_1, r) = 0$ if and only if $g(b_2, r) =0$. Therefore, $\cP(b_1) = \cP(b_2)$.
\end{proof}

\noindent \textbf{Proof of Lemma \ref{equivalenceClass}:} \label{third}
Let $\{x_t\}_{t\in \integers}$ represent the sequence of all the points in $\cS = \blues \cup \reds$, 
where $x_0 = b$, and for any $t\in \integers^+$, let $x_t$ and $x_{-t}$ represent the $t^{\text{th}}$ 
point to the right and the $t^{\text{th}}$ to the left of $b$, respectively. Since, $\cS$ is the mixture 
of two Poisson processes with rates $\lambda$ and $\mu$, $\{x_t\}$ occurs according to a Poisson 
process with rate $\lambda+\mu$. Moreover for any $t\in \integers$, $x_t$ is red with probability 
$p = \mu/(\lambda+\mu)$ and is blue with probability $q = 1-p = \lambda/(\lambda + \mu)$, 
independent from the color of the other points. 
Note that since $\lambda < \mu$, $p>1/2$.
Define the sequence $\{g_t\}_{t\in \integers^+}$, where $g_t = g(x_{-1}, x_t)$. It is easy to see that,
\begin{align}
g_{t+1} =
\left\{
	\begin{array}{ll}
		g_{t}+1  & \mbox{if $x_t$ is red}  \\
		g_{t}-1 & \mbox{if $x_t$ is blue}
	\end{array}.
\right.
\nonumber
\end{align}
Therefore, the sequence $\{g_t\}$ is equivalent to a random walk on 
$\integers$ that starts at $g_1 = -1$ and moves to the right or left according to probabilities 
$\mathbb{P}(+1) = p$ and $\mathbb{P}(-1) =q $. Suppose this random walk hits $0$ at some 
some time $t \in \integers^+$, {\it i.e.}\ , $g_t = 0$. If $x_{t-1}$ is a red point then since
$g(b, x_{t-1}) = g(x_{-1}, x_t) = g_t = 0$, then $x_{t-1} \in \cP(b)$. If $x_{t}$ is a blue point then since 
$g(b, x_{t}) = g(x_{-1}, x_t) + 1= g_t + 1 = 1$, then $x_{t} \in [b]$. A geometric interpretation of these facts is the following: If the random walk at time $t$ hits zero from below, then 
$x_{t-1} \in \cP(b)$, and if it hits zero from above and then goes below zero, then $x_t \in [b]$. \\
\textbf{Part (i)}.
Since $p>1/2$, $\lim_{t \rightarrow +\infty} g_t = +\infty$, almost surely. Therefore, since 
the random walk starts at $-1$, it hits $0$ at some time $t>0$, almost surely. 
Let $t_1$ be the first time that the random walk hits $0$, {\it i.e.}\ , $t_1 = \min\{t: g_t = 0\}$. 
Define $r_1 = x_{t_1-1}$. According to what we discussed above, $ r_1 \in \cP(b)$. 
Similarly, we can prove the existence of $r_2$. \\
\textbf{Part (ii)},
Let $b_1 = \min \{x\in [b]: x>b\}$. 
Therefore, $b$ and $b_1$ are consecutive points in $[b]$. It is clear that it is sufficient to prove the 
statement for these two points. According to lemma 
\ref{equivalenceClass}, $g(b, b_1) = 1$. Therefore, $g(x_{-1}, b_1) = g(b, b_1) -1 = 0$.
Since $t_1$ (defined in part (i)) is the first time that the random walk hits $0$, then $r_1 < x_{t_1}\le b_1$. 
Therefore, $r_1$ is a potential match between $b$ and $b_1$. If there exists another
potential match $r_2 = x_{t_2} \in \cP(b)$ between $b$ and $b_1$, then since $g_{t_2+1} = 0$,
$r_2 > r_1$. But $g_{t_1} = g_{t_2+1} = 0$ and  $g_{t_2} = -1$. Therefore, there exists
$t_1\le t' < t_2$ such that $g_{t'} = 0$ and $g_{t'+1} = -1$. Therefore, $b' = x_{t'} \in [b]$
which is a contradiction with the fact hat $b$ and $b_1$ are consecutive in [b]. \\
\textbf{Part (iii)}. The third part of the lemma is an immediate result of part (ii). \\
\textbf{Part (iv)} Following the discussion we had at the beginning, the random walk
finds a new blue point in $[b]$ if and only if it hits zero and then goes to $-1$. 
However, starting from zero, the probability that the random walk visits $-1$ again is 
$q/p = \lambda/\mu$ \cite{kelly2014stochastic}. This proves that $N^+(b)$ has a geometric distribution with rate 
$1-\lambda/\mu$. By symmetry, the same holds for $N^-(b)$.
\epr
\begin{lemma} \label{busycycle}
Let $\cW = \{w_i\}$ be a potential wave. Define $u_i = w_i - w_{i-1}$ (for valid
values of $i$). Then $u_i$'s are i.i.d.\ random variables with distribution $B(\lambda, \mu)$.
\end{lemma}

\begin{proof}
Due to the existing symmetry, without loss of generality, assume $b = w_i$ is a blue 
point and $r=w_{i+1}$ is a red point. Now if we consider the forward queue matching described 
in section $2$, it is easy to see that $u_i = r - b$ is equivalent to the busy cycle of the 
corresponding LCFS-PR queue.
\end{proof}

\noindent \textbf{Proof of Theorem \ref{bound}:} \label{boundproof}
According to Lemma \ref{stableisnested}, the stable matching $\cM_s$ is nested. 
Let $b \in \blues$ be an arbitrary blue point and let $N = N^+(b)$ and $N' = N^-(b)$ . 
According to Lemma \ref{equivalenceClass}, $N$ and $N'$ are i.i.d.\ geometric random 
variables with parameter $1-{\lambda}/{\mu}$. Now consider $b$'s potential wave 
$\mathcal{W} = \{w_i\}$, $-2N'+1 \le i \le 2N-1$ and define 
$B_i  = w_i - w_{i-1}$ for $i \in \{1, \dots, 2N+1\}$, and $B_i'  = w_i - w_{i-1}$ 
for $i \in \{ -2N', \dots, 0\}$. According to Lemma \ref{busycycle}, $B_i$ and $B_i'$ are i.i.d.\ 
random variables with distribution $B(\lambda, \mu)$. Now since $b$ should be matched with 
one of its potential partners in $\cP(b)$, its matching distance cannot be less than 
$\min\{B_1, B_1'\}$ or more than 
$\max \left\{ \sum_{i=1}^{2N-1} B_i , \sum_{i=1}^{2N'-1} B_i' \right\}$.
\epr

\noindent \textbf{Proof of Corollary \ref{upperbound}:} \label{upperboundproof}
Let $b$ be an arbitrary blue point and let $\cW(b) = {w_i}$, $-2N'+1\le i \le 2N-1$, represent 
its potential wave. Define $u_i = w_{i+1} - w_{i}$ for $-2N'+1\le i \le 2N-2$. Let $c_i$ represent 
the number of matching segments in $\cM_s$ which contain the interval $(w_{i}, w_{i+1})$. 
It is clear that $c_i$ cannot be larger than the total number of points in $\cW(b)$ to the left 
of $w_{i+1}$ or the total number of points in $\cW(b)$ to the right of $w_{i}$. In other words, if we define $ a_i = \min(i+2N',2N-1-i)$ then $c_i \le a_i$. Therefore,
\begin{align}
\Expect\bigl(d_\cM(b) \big{|} |[b]|\bigr) = 
\frac{1}{|[b]|} \sum_{b'\in [b]} \Expect\bigl(d_\cM(b') \big{|} |[b]\bigr) 
\nonumber 
= \frac{1}{|[b]|} \sum_{i = -2N'+1}^{2N-2}  c_i\Expect \left(u_i\right) \nonumber 
&\le \frac{1}{|[b]|} \sum_{i = -2N'+1}^{2N-2} a_i  \frac{1}{\mu - \lambda} \nonumber \\
& = (|[b]|+1) \frac{1}{\mu-\lambda}. \nonumber
\end{align}
Since $\Expect\left([b]\right) = \Expect(N) + \Expect(N') - 1 = 2\mu/(\mu-\lambda)-1$, the proof is complete.
\epr

\subsection{\textbf{Proof of Theorem \ref{Xupperbound}}} \label{forth}
In this section we prove Theorem \ref{Xupperbound}.
Suppose $n  = 2m >0$ be an even positive integer and let $\mathbf{x} = (x_1, ..., x_n) $ 
be a sequence of $n$ positive numbers. 
Let $\mathcal{C}(\mathbf{x})$ represent the configuration of $n+1$ red and blue points 
$P_1, ..., P_{n+1}$ which are placed in order on the real line such that the point $P_1$ is 
at the origin, and for any $1 \le i \le n$, $P_{i+1} = P_{i} + x_i$. 
Moreover, suppose $P_i$ is blue if $i$ is even and it is red if $i$ is odd. 
Therefore, there are $m+1$ red points and $m$ blue points in $\mathcal{C}(\xx)$.
Since there is no assumption on the value of $x_i$'s, the stable matching between points in 
$\blues$ and $\reds$ is not necessarily unique. (as an example consider the case where
 $x_i = 1$). Let $\cM$ be a stable matching for this problem which is generated 
 by the following algorithm. The algorithm repeatedly matches an unmatched blue point $b$ and an
 unmatched red point $r$ which have the minimum distance from each other, among all the
 remaining unmatched points, till no further matching is possible. 
 Let $D_\cM(x)$ represent the sum of all matching distances in 
$\mathcal{M}$,
 $$D_\cM(x) \triangleq \sum_{P \in \mathcal{C}(\mathbf{x}) \text{ is blue}} 
 |P- \cM(P)|.$$
Define $D(\xx)$ to be the expected value of $D_\cM(x)$, where the expectation is taken with
respect to the random stable matchings $\cM$ which is generated according to the algorithm.
Let $\prod_{n}$ represent the set of all permutations of n items. For any $\pi \in \prod_n$, let
$D_\pi(\xx) = D\bigl(\pi(\xx)\bigr) = D\bigl(\pi(x_1), ..., \pi(x_n)\bigr)$. Define $E(\mathbf{x})$ to be the 
expected value of $D_\pi(\xx)$, where the expectation is taken with respect to the random permutation 
$\pi$:
$$E(\xx) = \frac{1}{n!} \sum_{\pi \in \Pi_n} D_\pi(\xx).$$
Figure 1 shows $\mathcal{C}(\mathbf{x})$ for $n = 4$ and $\mathbf{x} = (.1, .2, .3, .4)$. 
We can easily see that here $D(\xx) = .4$ and $E(\xx) = {11}/{30}$. 
We will show later that $E(\xx)$ is a concave function of $\xx$. First consider the following definitions.

\begin{figure}[t]
\begin{center}
\includegraphics[width=0.48\textwidth]{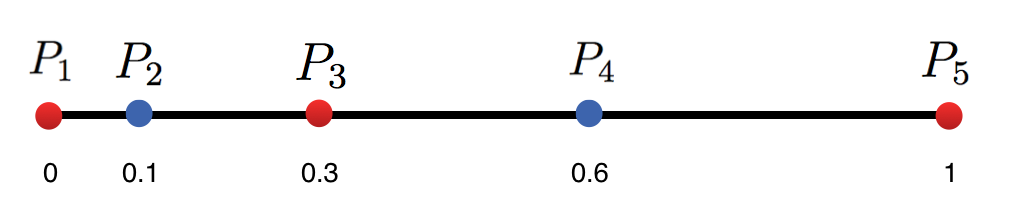}
\end{center}
\caption{$\mathcal{C}(.1, .2, .3, .4)$}
\label{f-nopsfrag}
\end{figure}

\begin{definition} 
For any positive integer $k$, Define $\mathcal{O}_k$ to be
the set of all sets of odd positive integers $\mathbf{o} = \{o_1, ..., o_r\}$ where $\sum_{i=1}^r{o_i} = k$.
\end{definition}
\begin{definition} \label{def2}
Let $\mathbf{o} = \{o_1, ..., o_r\} \in \mathcal{O}_k$. Define the function 
$P: \mathcal{O}_k \rightarrow \reals^+$ as follows:
$$P(\oo) \triangleq \frac{2^{r-1}} {\prod_{i=1}^k i^{c(\mathbf{o}, i)} c(\mathbf{o},i)!},$$
where $c(\mathbf{o}, i)$ represents the number of appearance of the number $i$ in $\mathbf{o}$.
\end{definition}
\noindent
For example, for $k=6$ we have $\mathcal{O}_6 = \left\{ \{1,5\}, \{3,3\}, \{3,1,1,1\}, \{1,1,1,1,1,1\}\right\}$.
Furthermore, $P\left(\{1,5\}\right) = \frac{2}{5}$, $P\left(\{3,3\}\right) = \frac{1}{9}$,
$P\left(\{3,1,1,1\}\right) = \frac{4}{9}$, and $P\left(\{1,1,1,1,1,1\}\right) = \frac{2}{45}$. In the following
lemma, we prove that for any $k$ the function $P$ is a probability measure on $\mathcal{O}_k$.

\begin{lemma} The function $P$ is a probability measure on $\mathcal{O}_k$. In other words:
$$\sum_{\mathbf{o} \in \mathcal{O}_k} P(\mathbf{o}) = 1.$$
\end{lemma}
\begin{proof}
From Taylor series expansion of $\log(x)$, for $|x| < 1$ we have,
$$ \log(1+x) - \log(1-x)  = 2 \sum_{i = 1}^\infty \frac{x^{2i+1}}{2i+1}.$$
Therefore,
$$ \Rightarrow \frac{1+x}{1-x} = \prod_{i=1}^\infty \exp\left(2\frac{x^{2i+1}}{2i+1} \right).$$
Now using the Taylor expansion of $e^x$, we have:
$$ 1 + \sum_{k=2}^\infty 2x^k = \frac{1+x}{1-x} = \prod_{i=1}^\infty \sum_{j = 0}^\infty \frac{\left(2x^{2i+1}\right)^j}{(2i+1)^j j!}.$$
From definition \ref{def2},  it is easy to see that if we expand the right hand side of the above equation 
as $\sum_{k=0}^\infty a_k x^k$, then for any $k\ge1$, $a_k = 2\sum_{\oo \in \mathcal{O}_k} P(\oo)$. 
In Other words,
$$ 1 + \sum_{k=1}^\infty 2x^k = \prod_{i=1}^\infty \sum_{j = 0}^\infty \frac{\left(2x^{2i+1}\right)^j}{(2i+1)^j j!} = 1 + \sum_{k=1} \left( 2\sum_{\oo \in \mathcal{O}_k} P(\oo)\right) x^k.$$
This proves that for any $k\ge 1$,
\begin{equation*}
\sum_{\oo \in \mathcal{O}_k} P(\mathbf{o}) = 1. 
\end{equation*}
\end{proof}

\begin{definition}
For any $\mathbf{o} = \{o_1, \dots, o_r\} \in \mathcal{O}_k$ and any $\xx = \{x_1, \dots, x_n\}$, define $\Phi(\mathbf{o}, \xx)$ to be the set of all possible ways of
splitting $\{x_1, \dots x_n\}$ into sets with sizes $o_1, \dots, o_r$. 
\end{definition}
\noindent
For example for $\mathbf{o} = \{1,1\} \in \mathcal{O}_2$ and $\xx = \{x_1, x_2, x_3, x_4\}$,
$$\Phi(\mathbf{o}, \xx) = \big\{ \{x_1,x_2\}, \{x_1,x_3\}, \{x_1,x_4\}, \{x_2,x_3\}, \{x_2,x_4\},
\{x_3,x_4\}\big\}.$$
It's easy to see that 
\begin{align} \label{eq0}
|\Phi(\mathbf{o}, \xx)| = \frac{n!}{(n-k)!\prod_{i=1}^r o_i! \prod_{i=1}^k c(\mathbf{o}, i)!},
\end{align}
where $c(\mathbf{o}, i)$ is the number of appreance of the number $i$ in $\oo$.

\begin{definition} \label{deff}
For any $\mathbf{o} = \{o_1, \dots, o_r\} \in \mathcal{O}_k$ and any $\xx = \{x_1, \dots, x_n\}$, define $f(\mathbf{o}, \xx)$ as follows:
$$f(\mathbf{o}, \xx) = \frac{1}{|\Phi(\mathbf{o}, \xx)|} \sum_{\phi \in \Phi(\mathbf{o}, \xx)}
\min \left( \sum_{i=1}^{o_1} x_i^{(1)}, \dots, \sum_{i = 1}^{o_r} x_i^{(r)}\right),$$
where in the summation it is assumed that $\phi$ is in the following form,
$$\phi = \left\{ \{x_1^{(1)}, \dots x_{o_1}^{(1)}\}, \dots 
, \{x_1^{(r)}, \dots x_{o_r}^{(r)}\} \right\}.$$
\end{definition}
\noindent
For example for $\mathbf{o} = \{1,1\} \in \mathcal{O}_2$ and $\xx = \{x_1, x_2, x_3, x_4\}$,
$$f(\mathbf{o}, \xx) = \frac{1}{6} \bigl( \min(x_1, x_2), \min(x_1, x_3), \min(x_1, x_4),
\min(x_2, x_3), \min(x_2, x_4), \min(x_3, x_4) \bigr).$$
In the following theorem, we find a close expression for $E(\xx)$ based on 
$f(\oo_k, x)$ for different values of $\oo_k$. 
\begin{theorem} \label{Ex}
For any $\xx \in \mathcal{S}_n$,
$$E(\xx) = \sum_{\substack{k\le n \\ k\text{~is even} \\ \oo_k \in \mathcal{O}_k}} P(\oo_k)f(\oo_k, \xx),$$
\end{theorem}

\begin{proof}
Let $n = 2m$. Suppose $\xx = (x_1, \dots, x_n)$ and without loss of generality assume 
that $x_1 \le x_2 \le \dots \le x_n$. Let $\pi$ be a random permutation which is drawn uniformly 
from $\prod_n$. In order to find a stable matching for $\pi(\xx)$, the algorithm should first find an index 
$i_1 \in \{1, \dots, n\}$ where $P_{i+1} - P_{i} = x_1$, and then matches points $P_i$ and $P_{i+1}$
to each other. Since, $\pi$ is drawn uniformly at random, $i_1$ is uniformly distributed over 
 $\{1, \dots, n\}$, {\it i.e.}\ , $\Prob(i_1 = i) = 1/n, \forall i \in \{1, \dots, n\}$. 
 Due to the existing symmetry we can assume that $i_1 \in \{1, \dots, m\}$. Now, conditioning
 on $i_1$ we have,
 \begin{align} \label{eq1}
 E(\xx) = \Expect\bigl(D_\pi(\xx)\bigr) = \frac{1}{m} \sum_{i = 1}^m \Expect\bigl(D_\pi(\xx) | i_1 = i\bigr).\end{align}
 For $i = 1$, by conditioning on the second segment at index $2$, we have
 \begin{align} \label{eq2}
 \Expect\bigl(D_\pi(\xx) | i_1 = 1\bigr) = x_1 + \frac{1}{n-1}\sum_{2\le i\le n}  E(\xx_{\{{i}\}}),
 \end{align} 
 where $\xx_{\{i\}}$ represents the sequence of $n-2$ positive numbers which is generated from
 $\xx$ by removing $x_1$ and $x_i$.
For $i>1$, by conditioning on the segments at indices $i-1$ and $i+1$, we have

\begin{align} \label{eq3}
 \Expect\bigl(D_\pi(\xx) | i_1 > 1\bigr) = x_1 + \frac{2}{(n-1)(n-2)}\sum_{2\le i < j \le n}^n  E(\xx^{\{{i}, {j}\}}),
\end{align}
 where $\xx^{\{{i}, {j}\}}$ represent the sequence of $n-2$ positive integers which is generated
 from $\xx$ by removing $x_1$, $x_i$, and $x_j$ and adding $x_1 + x_i + x_j$. We prove the
 theorem by induction on $n$. For $n = 2$, it is clear that $E\bigl( (x_1, x_2)\bigr) = \min(x_1, x_2)$.
 Assume the theorems statement is valid for n-2. From equations (\ref{eq1}), (\ref{eq2}), and (\ref{eq3}) 
 and from the induction assumption it is clear that $E(\xx)$ can be written as a linear combination 
 of terms with the following form, 
 \begin{align} \label{eq4}
 A = \min \left( \sum_{x \in S_1} x, \dots, \sum_{x \in S_r} x\right),
 \end{align}
 where $S_i$'s are disjoint subsets of $\{x_1, \dots, x_n\}$ with an odd size. Suppose
 $S_i =  \{x_1^{(i)}, \dots x_{o_i}^{(i)}\} $ where $o_i = |S_i|$ and let $k = \sum_i o_i$. Define
 $S = \cup S_i$.
 Let $C$ represent the constant factor of the the term $A$ show in equation (\ref{eq4}) in $E(\xx)$. 
According to equation \ref{eq0}, It is sufficient to prove that, 
 \begin{align} \label{Ceq}
C &= P\bigl( (o_1, \dots, o_k) \bigr) \times 
\frac{{(n-k)!\prod_{i=1}^r o_i! \prod_{i=1}^k c(\mathbf{o}, i)!}}{n!} \nonumber\\
& = \frac{2^{r-1}} {\prod_{i=1}^k i^{c(\mathbf{o}, i)} c(\mathbf{o},i)!} \times
\frac{(n-k)!{\prod_{i=1}^r o_i! \prod_{i=1}^k c(\mathbf{o}, i)!}}{n!} \nonumber\\
& = \frac{(n-k)!2^{r-1}\prod_{i=1}^r o_i!}{n! \prod_{i=1}^r o_i}.
\end{align}
We Prove this by considering two different cases: (i) $x_1 \notin S$, and 
(ii) $x_1 \in S_i$ for some $i$.

\begin{enumerate}[label=(\roman*)]
\item The term $A$, has the following constant factor $C_1$ in the equation \ref{eq2}, 
\begin{align}\label{C11}
C_1 &= \frac{n-k-1}{n-1}\frac{(n-2-k)!2^{r-1}\prod_{i=1}^r o_i!}{(n-2)! \prod_{i=1}^r o_i} \nonumber \\
& = \frac{n}{n-k}\frac{(n-k)!2^{r-1}\prod_{i=1}^r o_i!}{n! \prod_{i=1}^r o_i}. 
\end{align} 
The reason is that in order for term to appear in in $E(\xx_{\{i\}})$, it is sufficient to have $i \notin S$ which
occurs with probability $\frac{n-k-1}{n-1}$.
The term $A$ also appear in $E(x^{\{i,j\}})$ if $i \notin S$ and $j \notin S$ which occurs with 
probability $\frac{(n-k-1)(n-k-2)}{(n-1)(n-2)}$. Therefore, the term $A$ has the following constant 
factor $C_2$ in the equation \ref{eq3},
\begin{align}\label{C21}
C_2 &= \frac{(n-k-1)(n-k-2)}{(n-1)(n-2)}\frac{(n-2-k)!2^{r-1}\prod_{i=1}^r o_i!}
{(n-2)! \prod_{i=1}^r o_i} \nonumber \\
& = \frac{n(n-k-2)}{(n-k)(n-2)}\frac{(n-k)!2^{r-1}\prod_{i=1}^r o_i!}{n! \prod_{i=1}^r o_i}. 
\end{align}
Since, $C = \frac{2}{n}C_1 + \frac{n-2}{n}C_2$, it is easy to derive equation (\ref{Ceq}) from equations (\ref{C11}) and (\ref{C21}).
\item Without loss of generality assume $x_1\in S_1$. First we consider the case where $o_1 > 1$, {\it i.e.}\ , 
$|S_1| \ge 3$. It is clear that in this case the term $A$ does not appear in equation (\ref{eq2}), {\it i.e.}\ ,
$C_1 = 0$. $A$ appears in $E(\xx^{\{i,j\}})$ if and only if $i, j \in S_1$ which occurs with probability 
$\frac{(o_1-1)(o_1-2)}{(n-1)(n-2)}$. 
Therefore, it has the following constant factor $C_2$ in the equation (\ref{eq3}),
\begin{align}\label{C22}
C_2 & = \frac{(o_1-1)(o_1-2)}{(n-1)(n-2)}
\frac{(n-k)!2^{r-1}(o_1-2)!\prod_{i=2}^r o_i!}{(n-2)! (o_1 - 2)\prod_{i=2}^r o_i}.  \nonumber \\
& =\frac{n}{n-2} \frac{(n-k)!2^{r-1}\prod_{i=1}^r o_i!}{n! \prod_{i=1}^r o_i}. 
\end{align}
Since $C = \frac{n-2}{n}C_2$, we can derive the equation (\ref{Ceq}) from (\ref{C22}).
For $o_1 = 1$, the value of $A$ is equal to $x_1$. On the other hand if we consider all
the terms $A$ with $S_1 = \{x_1\}$, it is easy to see that the sum of all of their coefficient
is equal to $1$ which is consistent with the coefficient of $x_1$ in equation \ref{eq1} (if we
plug in equations (\ref{eq2})) and (\ref{eq3}).
\end{enumerate}
 \end{proof}
\begin{corollary}\label{concave} 
$E(\xx)$ is a concave function of $\xx$. 
\end{corollary}
\begin{proof}
According to definition \ref{deff} it is clear that $f(\oo, \xx)$ is a concave function of $\xx$ (for any
$\oo$). Therefore, from Theorem \ref{Ex} we can conclude that $E(\xx)$ is a concave function 
of $\xx$.
\end{proof}

Now define $\xx_n = (1, \dots, 1)$, to be the sequence of $n$ numbers all equal to $1$.
The following lemma proves an upper bound for $E(\xx)$ for arbitrary $\xx$ in terms of $E(\xx_n)$.
\begin{lemma} \label{xnupper}
For any $\xx = (x_1, \dots, x_n)$, we have
$$E(\xx) \le E(\xx_n) \bigl(\frac{1}{n}\sum_i x_i\bigr).$$
\end{lemma}
\begin{proof}
For any $i \in \{0, \dots, n-1\}$, define $\xx^{(i)} = (x_{i+1}, \dots, \xx_n, x_{1}, x_{i})$. First 
not that $E(\xx^{(i)}) = E(\xx), \forall i\in \{0, \dots, n-1\}$.Since
$E(\xx)$ is a concave function of $\xx$, we have
\begin{equation}
E(\xx) = \frac{1}{n}\sum_i E(\xx^{(i)}) \le E\bigl(\frac{1}{n}\sum_i \xx^{(i)}\bigr) 
=E\bigl( \frac{1}{n}\sum_i x_i, \frac{1}{n}\sum_i x_i, \dots, \frac{1}{n}\sum_i x_i\bigr) =  E(\xx_n) \bigl(\frac{1}{n}\sum_i x_i\bigr). \nonumber
\end{equation} 
\end{proof}

In the following lemma we prove an upper bound for $E(\xx_n)$.
\begin{lemma} \label{loglemma}
Let $n = 2m$. We have,
$$E(\xx_n) \le m(1+\ln m).$$ 
\end{lemma}
\begin{proof}
Similar to the proof of Theorem \ref{Ex} conditioning on the value of the first index $i_1$, we have
\begin{align}
E(\xx_n) = {1} + \frac{1}{m} E(\xx_{n-2}) + \frac{m-1}{m} E(\xx_{n-2}'),
\end{align}
where $\xx_{n-2}' = (3, 1, \dots, 1)$ and has $n-2$ entries. According to Lemma 
\ref{xnupper}, $E(\xx_{n-2}') \le \frac{n}{n-2}E(\xx_{n-2})$. Therefore,
\begin{align}
E(\xx_n) &\le {1} + (\frac{1}{m} + 1)E(\xx_{n-2}) \nonumber \\
& \le 1 + (\frac{m}{m-1})E(\xx_{n-2}). \nonumber
\end{align}
Therefore,
\begin{align}
\frac{E(\xx_n)}{m} \le \frac{1}{m} + \frac{E(\xx_{n-2})}{m-1}.
\end{align}
By induction on $m$, we have 
$$\frac{E(\xx_n)}{m} \le \frac{1}{m} + \frac{1}{m-1} + \dots + 1 \le 1 + \ln m.$$
This completes the proof.
\end{proof}

\noindent \textbf{Proof of theorem \ref{Xupperbound}:}
Let $b$ be an arbitrary blue point and let $\cW(b) = {w_i}$, $-2N'+1\le i \le 2N-1$, represent 
its potential wave. Define $u_i = w_{i+1} - w_{i}$ for $-2N'+1\le i \le 2N-2$. Let $m = |[b]|$.
Assume that $u_i$'s take their values from the set $\xx = \{x_1, \dots, x_{2m}\}$. According
to the definition of $E(\xx)$, we have
\begin{align}
\Expect\bigl(X | m, \xx\bigr) = \frac{1}{m}E(\xx) \le (1 + \ln m) \bigl(\frac{1}{2m} \sum_i x_i\bigr), \nonumber
\end{align}
where in the last inequality we used the results of Lemma \ref{xnupper} and \ref{loglemma}. 
Since $\xx$ is arbitrary, we can conclude
\begin{align}
\Expect(X | m)  = \Expect \bigl( \Expect(X| \xx, m)\bigr) \le (1 + \ln m) \frac{1}{\mu-\lambda}. \nonumber
\end{align}
On the other hand, since $\ln(x)$ is a concave function, from Jensen inequality we have
\begin{align}
\Expect(X) = \Expect\bigl(\Expect(X | m)\bigr) \le \Expect\bigl( (1 + \ln m) \bigr)  \frac{1}{\mu-\lambda}
\le \bigl( 1 + \ln \Expect(m)\bigr)  \frac{1}{\mu-\lambda} 
= \left(1 + \ln \bigl (\frac{\mu+\lambda}{\mu - \lambda} \bigr)\right) \frac{1}{\mu-\lambda}. \nonumber
\end{align}

\epr

\bibliography{bibliography}


\end{document}